\documentclass[letterpaper,journal]{IEEEtran}

\usepackage{cite}
\usepackage{amsmath,amssymb,amsfonts}
\usepackage{graphicx}
\usepackage{textcomp}
\usepackage{xcolor}
\usepackage{algorithm}
\usepackage{algpseudocode} % OR algorithmic, choose only one
\usepackage{booktabs}
\usepackage{multirow}
\usepackage{amsthm}
\usepackage{tikz}
\usepackage{pgfplots}
\pgfplotsset{compat=1.17}
\usepackage{balance}

\usepackage{hyperref}
\usepackage[capitalise]{cleveref}

\newtheorem{remark}{Remark}

\crefname{equation}{}{}
\Crefname{equation}{}{}

\crefname{figure}{Fig.}{Figs.}
\Crefname{figure}{Fig.}{Figs.}

\crefname{appendix}{Appendix}{Appendices}
\Crefname{appendix}{Appendix}{Appendices}

\crefname{assumption}{Assumption}{Assumptions}
\Crefname{assumption}{Assumption}{Assumptions}

\newtheorem{definition}{Definition}
\newtheorem{assumption}{Assumption}
\newtheorem{theorem}{Theorem}
\newtheorem{lemma}{Lemma}
\newtheorem{Corollary}{Corollary}
\title{
A Joint Power-Privacy Control Framework for Decentralized Learning over Heterogeneous Wireless Multicasting Networks
}

\author{
    Amir~Ziaeddini,~\IEEEmembership{Member,~IEEE}, 
    Yauhen~Yakimenka, 
    J{\"o}rg~Kliewer,~\IEEEmembership{Fellow,~IEEE}
    
    \thanks{
    Amir Ziaeddini, Yauhen Yakimenka and J{\"o}rg Kliewer are with the Helen and John C. Hartmann  Department of Electrical and Computer Engineering, New Jersey Institute of Technology (NJIT), Newark, NJ, USA
    (e-mail: az328@njit.edu; yauhen.yakimenka@njit.edu; jkliewer@njit.edu). 
    }
    
}

\begin{document}

\maketitle

% -------------------------------------------------------------
% ABSTRACT
% -------------------------------------------------------------
\begin{abstract}
In this paper, we propose a decentralized learning framework that incorporates both power control and privacy guarantees. Specifically, we enable a set of clients in a wireless multicast network to jointly train a common model while maintaining a prescribed per-iteration maximum privacy leakage level. The communication network is represented by a row-stochastic adjacency matrix, allowing us to capture asymmetric channel gains as well as heterogeneous maximum transmit power levels. 
%Differential privacy (DP) is applied concurrently with a power-sharing strategy that distributes maximum transmit power between the exchanged model coefficients and the injected Gaussian noise, enabling the learning process to achieve both convergence and privacy guarantees. 
Differential privacy is enforced through an explicit power-splitting strategy that allocates each node’s limited maximum transmit power between model coefficients and injected Gaussian noise, thereby jointly controlling learning performance and privacy leakage.
We further prove that the proposed algorithm achieves a cumulative regret bound of $O(\log T)$, where $T$ denotes the time horizon. To evaluate the practical performance of our approach, we perform comprehensive experiments on the CIFAR-10 dataset under both IID and non-IID data distributions, considering %a wide range of %realistic wireless settings. 
%The algorithm is tested under both IID and non-IID data heterogeneity, as well as %under 
different privacy levels, diverse numbers of clients, and various graph topologies. The results demonstrate strong performance across the considered settings and improved performance over existing methods, highlighting the effectiveness of the proposed algorithm under realistic wireless communication constraints.
\end{abstract}

% -------------------------------------------------------------
% INDEX TERMS
% -------------------------------------------------------------

% -------------------------------------------------------------
\section{Introduction}\label{sec:intro}
\IEEEPARstart{R}{ecently}, decentralized learning (DL) has gained significant attention in federated and collaborative AI systems \cite{10542323, 10251949, 10420449, 9220780}, driven by the rapid growth of data generated by mobile devices, IoT sensors, autonomous systems, and communication networks. Unlike centralized learning (CL), where the training process relies on a central server to coordinate and aggregate information from participating users \cite{9460016,9306745,abrar2025nonconvexovertheairheterogeneousfederated}, DL distributes both computation and coordination across the network. Rather than relying on a central entity, each node performs local computations and exchanges model parameters or gradients directly with its neighboring nodes. This peer-to-peer structure removes the central bottleneck and single point of failure associated with CL, making DL well suited to naturally distributed applications such as connected vehicles, wearable health monitoring, smart cities, and large-scale sensor networks. However, effective DL requires reliable communication over the underlying graph and learning algorithms that remain convergent and stable under heterogeneous and non-IID data distributions as well as noisy communication channels.

Wireless multicast communication is well suited to DL, as it enables simultaneous information exchange among multiple nodes over a shared wireless medium \cite{10902529,9014530,10025677}. In particular, over-the-air computation (OAC) \cite{9563232,9322286,10506083,11148660} exploits wireless superposition to efficiently aggregate model updates. However, practical wireless networks involve heterogeneous channel gains, asymmetric links, and limited transmit power, which can significantly affect the learning process.

Beyond these communication challenges, DL also faces privacy risks arising from the exchange of model information. In particular, exchanged model coefficients can potentially be exploited by adversaries to recover users' original data. Prior work has demonstrated that techniques such as model inversion or reconstruction attacks can extract sensitive features or even recover original data samples from the exchanged loss function gradients or model coefficients \cite{zhu2019deep,geiping2020inverting,10024757}. These risks become even higher when the clients have small datasets or when the data has clear and recognizable patterns. To address this privacy risk, differential privacy (DP) provides a way to safeguard individual users' original data \cite{dwork2014algorithmic,9714350,9069945}. The idea behind DP is to add carefully calibrated random noise to the parameters exchanged between clients such that no single data point can be distinguished from another. Integrating DP into DL combines the benefits of distributed computation with privacy protection, allowing the systems to learn effectively from large  datasets while maintaining user data privacy.

% \textcolor{blue}{Since DL systems operate over communication networks, the underlying wireless transmission environment plays a critical role in the learning process. Wireless multicast communication has emerged as an efficient approach for DL systems due to its ability to simultaneously deliver information over a shared wireless medium \cite{10902529,9014530,10025677}. In particular, over-the-air computation (OAC) \cite{9563232,9322286,10506083,11148660} exploits the wireless superposition property to aggregate model updates directly in the air, improving communication efficiency and reducing latency. However, practical wireless multicast networks introduce challenges such as heterogeneous channel gains, fading, interference, asymmetric links, and limited transmit power budgets, which naturally lead to non-uniform aggregation behavior and complicate DL and privacy preservation. Although wireless multicasting provides an attractive communication paradigm for decentralized learning, existing DP-DL algorithms have largely been developed under idealized network models that overlook the communication characteristics of practical wireless multicast systems.}

Most existing works on differentially private decentralized learning (DP-DL) rely on the assumption that the network communication graph is either balanced or undirected, which forces the corresponding adjacency (weight) matrix to be doubly-stochastic \cite{9716792, 10506083, 10279097,10025677, 10115431,rodio2025optimizingprivacyutilitytradeoffdecentralized}. Although this assumption greatly simplifies the analysis and makes convergence easier to establish, it does not reflect how real communication networks operate. In practical wireless networks, links are often asymmetric and experience different channel conditions and interference levels. Therefore, enforcing doubly-stochastic weights rarely matches the actual behavior of large, heterogeneous, and dynamic networks. To relax these limitations, prior works have introduced decentralized optimization methods that avoid the balanced graph requirement. The subgradient-push frameworks \cite{7405263,6426375,6930814} and push-sum-based algorithms \cite{8433217, 10068288, 11079243} enable learning over directed graphs by using a column-stochastic weight matrix. Push-sum–based approaches further generalize the concept by utilizing two complementary weight matrices, namely a row-stochastic and a column-stochastic matrix, to correct directional imbalance and guarantee consensus within push–pull schemes \cite{10535197,8988200,10337617}. While effective in theory, all of these methods require carefully designed weight assignments that are tightly coupled to the graph structure, making them difficult to implement in wireless communication settings where channel conditions vary. In OAC, however, wireless superposition and heterogeneous channel gains naturally determine the relative contributions of the received signals. Normalizing these channel-dependent coefficients at each receiving node yields a row-stochastic mixing structure, where the effective weights reflect both the transmit powers and wireless channel conditions. Motivated by such a structure, \cite{7526803, 8267245, 10124282} adopt row-stochastic weighting schemes in which each client autonomously determines the weights associated with its multicast transmissions to neighboring nodes. Despite this flexibility, these works do not account for privacy and neglect the potential leakage arising from such exchanges. In \cite{9013030} and \cite{10552083}, the authors introduce DP-DL by injecting Laplace noise into the exchanged updates and assume that mixing is governed by a row-stochastic adjacency matrix. While these studies demonstrate that DP can be combined with DL over unbalanced directed graphs, the treatment of the weight matrix remains abstract and detached from the actual communication environment. In particular, the underlying wireless aspects such as node-specific transmit power constraints and heterogeneous channel gains are not reflected in the model. As a result, these approaches do not provide any mechanism for allocating each node’s limited power budget between multicasting model coefficients and injecting DP noise, which ultimately prevents them from achieving an appropriate balance between learning accuracy and data leakage.

Overall, existing research exposes three important gaps. First, most DP-DL algorithms rely on balanced or undirected communication graphs that are difficult to realize in practical wireless multicast networks. Second, although several existing decentralized optimization methods for directed graphs incorporate DP, they adopt generic noise injection mechanisms without accounting for the characteristics of wireless communication. Third, existing approaches lack a systematic mechanism for controlling the privacy-accuracy trade-off through transmit power allocation. These limitations motivate the development of a DP-DL framework that jointly exploits wireless multicasting, supports DL over unbalanced networks, and explicitly optimizes transmit power allocation to achieve a desirable balance between learning performance and privacy.

The main contributions are summarized as follows:
 \begin{itemize}
  \item We introduce a wireless multicast DP-DL framework for directed and unbalanced networks, where the row-stochastic aggregation weights naturally arise from heterogeneous wireless channel gains and transmit power allocation rather than being explicitly designed.

  \item We propose a novel joint power-privacy control framework that optimizes transmit power allocation to balance learning performance and privacy leakage while maintaining a predefined DP level.
  
  \item We provide a theoretical analysis by establishing DP guarantees and deriving an $O(\log T)$ cumulative regret bound that characterizes the effects of power control, noise injection, and network characteristics on the learning process.

  \item We conduct comprehensive experiments on the CIFAR-10 dataset across a variety of
  DP-DL scenarios, including multiple graph structures (fully connected, ring,
  and random topologies), varying numbers of clients, both IID and non-IID data distributions, and different privacy levels.
\end{itemize}

Compared to our previous work in \cite{11195385}, this paper substantially broadens the scope of the proposed DP-DL framework by incorporating a more comprehensive wireless communication model, extending the theoretical development, and providing a more extensive empirical evaluation, including an analysis of the cumulative privacy leakage, and a comparative performance evaluation against existing methods.

The remainder of this paper is organized as follows. \Cref{sec:prelim} reviews the work background and preliminaries, \cref{sec:system-model} presents the system model, and \cref{sec:algorithm} introduces the proposed algorithm. \Cref{sec:theoretical-analysis} provides the theoretical analysis, \cref{sec:sim-results} presents the simulation results, and \cref{sec:conclusion} concludes the paper. Additionally, the convergence proof is provided in \cref{app:proof}.
% -------------------------------------------------------------
\section{Background and Preliminaries}\label{sec:prelim}
% --------------------------------------------------------------
\subsection{Wireless DL via OAC}
DL aims to solve a global optimization problem without relying on a central server. Instead, a group of $K$ nodes communicate only with their neighbors to collaboratively minimize a shared global loss. Each node $i$ holds a private dataset $\mathcal{D}_i$ and a corresponding local objective function $f_i:\mathbb{R}^m\rightarrow\mathbb{R}$ defined over a convex constraint set $\Omega \subseteq \mathbb{R}^m$. The global learning task is formulated as the following optimization problem:
\begin{equation}\label{eq:optim-problem}
\min_{\mathbf{x} \in \Omega} F(\mathbf{x}) 
= \min_{\mathbf{x} \in \Omega} \sum_{i=1}^{K} f_i(\mathbf{x}) ,
% = \min_{\mathbf{x} \in \Omega} \sum_{i=1}^{K} f_i(\mathbf{x};D_i).
\end{equation}
where $\mathbf{x}$ is the model vector and each node maintains its own
model coefficients, $\mathbf{x}_{i,t}$. Here, $F(\mathbf{x})$ denotes
the global objective function obtained by aggregating the local objective
functions $f_i(\mathbf{x})$ of all $K$ nodes. The goal is to obtain an optimal solution to the global learning problem, denoted by
$\mathbf{x}^{*} = \arg\min_{\mathbf{x}\in\Omega} F(\mathbf{x})$.

The classical decentralized stochastic gradient descent (DSGD) algorithm addresses this optimization problem by combining gradient descent with consensus averaging. At iteration $t$, each node $i$ computes its stochastic gradient using a mini-batch $\mathcal B_{i,t} \subseteq \mathcal{D}_i$:
\[
 \mathbf{g}_{i,t} \triangleq \frac{1}{|\mathcal{B}_{i,t}|} \sum_{\xi\in \mathcal{B}_{i,t}} \nabla \ell(\mathbf{x}_{i,t};\xi),
\]
where $\ell(\mathbf{x}_{i,t}; \xi)$ denotes the sample-wise loss evaluated at the model coefficients $\mathbf{x}_{i,t}$, and data point $\xi \in \mathcal{B}_{i,t}$.

Nodes then perform two steps:
\begin{enumerate}
\item \textbf{Consensus mixing}:
\begin{equation}\label{eq:step-mixing}
\mathbf{y}_{i,t} = \sum_{j=1}^{K} a_{ij} \mathbf{x}_{j,t},
\end{equation}
where $A=[a_{ij}]$ is the mixing matrix compatible with the communication graph.

\item \textbf{Local update}:
\begin{equation}\label{eq:step-local-update}
\mathbf{x}_{i,t+1} = \mathbf{y}_{i,t} - \gamma_t  \mathbf{g}_{i,t},    
\end{equation}
where $\gamma_t$ is the learning rate.
\end{enumerate}

DSGD enables information to propagate across the network while each node simultaneously updates its model using its own local gradient. In connected graphs, this enables the entire set of nodes to collaboratively approximate the global optimum. 

In DL, the mixing matrix $A$ in \cref{eq:step-mixing} determines how nodes combine information from their neighbors. Under standard convergence conditions, when $A$ is doubly-stochastic (each row and column sums to one), the standard DSGD algorithm converges to the minimizer of the global objective in \cref{eq:optim-problem}. However, as we will describe in subsequent sections, our framework employs a row-stochastic adjacency matrix $A$ to capture the inherent asymmetry of wireless channel coefficients. In this case, the standard DSGD update rule in \cref{eq:step-local-update} becomes biased and no longer converges to the minimizer of $F(\mathbf{x})$ \cite{8316938}. Instead, it converges to the minimizer of a weighted objective $\bar{F}(\mathbf{x}) = \sum_{i=1}^K \pi_i f_i(\mathbf{x})$, where $\boldsymbol{\pi}=[\pi_1,\ldots,\pi_K]^\top$ is the left Perron eigenvector of $A$, satisfying $\boldsymbol{\pi}^\top A = \boldsymbol{\pi}^\top$. Because $\boldsymbol{\pi}$ is non-uniform, different nodes contribute unequally to the final solution, whereas in the doubly-stochastic case $\boldsymbol{\pi}$ becomes the uniform vector with entries $1/K$, giving each node equal influence. To remove this bias, the left Perron eigenvector is incorporated into \cref{eq:step-local-update} by scaling each node's gradient by $1/\pi_i$. The modified update is
\begin{equation}\label{eq:modified-update}
    \mathbf{x}_{i,t+1}
    = \sum_{j=1}^K a_{ij} \mathbf{x}_{j,t} - \frac{\gamma_t}{\pi_i}\, \mathbf{g}_{i,t}.
\end{equation}

This modification compensates for the imbalance caused by row-stochastic mixing and ensures that the algorithm converges to the minimizer of the global objective $F(\mathbf{x})$ \cite{8316938}.

In wireless networks, the consensus mixing in \cref{eq:step-mixing} can be implemented more efficiently using OAC. Instead of exchanging model parameters via standard point-to-point communication, OAC exploits the natural superposition of the wireless multiple access channel (MAC). Concurrent transmissions from multiple nodes are combined in the air, enabling the receiver to directly observe an analog sum of the transmitted signals. When multiple nodes transmit simultaneously, the $i$-th receiver directly obtains an analog sum of the signals {\cite{8952884,11131470,11048956},
\begin{equation*}
\sum_{j \in \mathcal{N}(i)} h_{ji} \sqrt{P_j} \mathbf{x}_{j,t} + \mathbf{n}_{i,t},    
\end{equation*}
where $h_{ji}$ is the wireless channel coefficient, $P_j$ is the transmit power, and $\mathbf{n}_{i,t}$ is Gaussian noise. The precise definitions of the communication graph, neighborhood set $\mathcal{N}(i)$, channel coefficients, transmit power constraints, and noise will be formally introduced in \cref{sec:system-model}.
% \textcolor{blue}{By exploiting this aggregated signal, OAC enables each node to incorporate information from its neighbors directly into the local learning update. Thus, the weighted mixing operation required by DSGD can be efficiently realized over the wireless medium through OAC. Through this integration, the network can collaboratively minimize the global objective and approach the optimal solution in a decentralized manner.}
\subsection{Differential Privacy}
When DL algorithms require nodes to exchange model coefficients, there is a risk that an adversary observing these exchanges may infer or recover sensitive details about the local datasets. DP provides a tool to mitigate such leakage by injecting carefully calibrated randomness into the transmitted
signals. In our decentralized setting, this typically means perturbing the transmitted model coefficients, ensuring that the influence of any single data sample remains statistically indistinguishable.

DP guarantees are expressed in terms of two parameters that quantify the allowable
information leakage. To describe these guarantees, we introduce the standard
definitions below.
\begin{definition}[\!\cite{9714350}]\label{def:DP}
A randomized mechanism $\mathcal{M}$ is said to satisfy $(\epsilon,\delta)$-DP 
if for every pair of neighboring datasets $D$ and $D'$ differing in exactly one entry,
and for every subset of outputs $\mathcal{O}$, the following inequality holds:
\begin{equation*}
    \Pr[\mathcal{M}(D) \in \mathcal{O}]
    \;\le\;
    e^{\epsilon}\,\Pr[\mathcal{M}(D') \in \mathcal{O}] + \delta.
\end{equation*}
\end{definition}

Here, $\epsilon$ represents the privacy budget or leakage, while $\delta$ is the failure probability, denoting the small probability that the privacy guarantee may fail to hold. A smaller value of $\epsilon$ provides stronger privacy protection, as the outputs of the mechanism 
become more indistinguishable across neighboring datasets. Conversely, increasing $\epsilon$ weakens 
privacy by allowing more information leakage.

A key quantity in DP analysis is the sensitivity of the function being perturbed.
It measures the maximum change in the output when only one element of the input dataset
is modified. Sensitivity determines the required noise to preserve privacy.

\begin{definition}[\!\cite{9714350}]\label{def:sensitivity}
The  $\ell_2$-sensitivity of a function $f$ is
\begin{equation}\label{eq:sensitivity-def}
    \Delta_f
    \;=\;
    \max_{D , D'} 
    \big\| f(D) - f(D') \big\|_2,
\end{equation}
where $D$ and $D'$ are neighboring datasets that differ in exactly one sample.
\end{definition}

To enforce DP, one common approach is the Gaussian mechanism \cite{9714350},
which adds Gaussian noise with variance proportional to the sensitivity of the function.
The resulting mechanism ensures that outputs from neighboring datasets remain
indistinguishable up to the chosen privacy parameters.

\textbf{Gaussian Mechanism}\cite{9714350}:
Let $f$ be a function with sensitivity $\Delta_f$. The Gaussian mechanism releases
\begin{equation*}
    \mathcal{M}(D) 
    \;=\;
    f(D) + \mathcal{N}\!\left(0,\sigma^2 \mathbf{I}\right),
\end{equation*}
where the standard deviation $\sigma$ required to achieve $(\epsilon,\delta)$-DP is given by
\begin{equation}\label{eq:sigma-DP}
    \sigma
    \;=\;
    \frac{\Delta_f}{\epsilon}
    \sqrt{2\ln\!\left(\frac{1.25}{\delta}\right)}.
\end{equation}

In DP-DL, this noise may be applied to individual gradients or transmitted model coefficients. By calibrating the noise according to the sensitivity and desired privacy budget, each participating node ensures that its local data remains protected across all communication rounds, while still enabling
collaborative model training. In this work, DP guarantees are provided at the level of each transmitted model update and each communication round.

Throughout the paper, $\Pi_{\Omega}(\cdot)$ denotes the Euclidean projection onto the constraint set $\Omega$ \cite{9013030}, defined as
\[
\Pi_{\Omega}(\mathbf{x}) \triangleq \arg\min_{\mathbf{y} \in \Omega} \|\mathbf{x} - \mathbf{y}\|.
\]
This projection ensures that each update $\mathbf{x}_{i,t+1}$ remains feasible. Since $\Omega$ is a nonempty closed convex set, $\Pi_{\Omega}(\cdot)$ is non-expansive (i.e., $1$-Lipschitz), implying that for any $\mathbf{u}, \mathbf{v}$,
\[
\big\|\Pi_{\Omega}(\mathbf{u}) - \Pi_{\Omega}(\mathbf{v})\big\| \le \|\mathbf{u} - \mathbf{v}\|.
\]

\section{System Model}\label{sec:system-model}
In this work, we study a wireless DP-DL system in which $K$ spatially distributed nodes collaborate to train a common global model without relying on the central server. Each node $i \in \{1,\ldots,K\}$ maintains its own private dataset $\mathcal{D}_i$ and computes local model coefficients that are exchanged only with its neighboring nodes, ensuring that raw data never leaves the originating device. Each node is equipped with two antennas, allowing full-duplex operation so that model parameters can be transmitted and received simultaneously, although each node is subject to a limited transmit power budget due to device hardware or battery constraints. We consider honest-but-curious neighboring nodes that may infer private information from received model updates.

This decentralized system is modeled as an unbalanced directed graph $\mathcal{G} = (\mathcal{V}, \mathcal{E})$, where $\mathcal{V} = \{1, 2, \ldots, K\}$ denotes the set of participating nodes and $\mathcal{E} \subseteq \mathcal{V} \times \mathcal{V}$ represents the set of directed communication links. The graph $\mathcal{G}$ is assumed to be strongly connected and time-invariant, implying that for every pair of nodes in the network, there exists at least one directed path that connects them. A directed edge $(i, j) \in \mathcal{E}$ indicates that node $i$ can transmit its information directly to node $j$, making $j$ an outgoing neighbor of $i$. For analytical convenience, we assume that every communication link is bidirectional, i.e., $(i,j)\in\mathcal{E}$ implies $(j,i)\in\mathcal{E}$, allowing neighboring nodes to communicate in both directions, while the corresponding directed edges may have different weights. Hence, the total incoming and outgoing weights at each node need not be equal, and the resulting communication graph can be unbalanced. The set of all neighbors of node $i$ is denoted by $\mathcal{N}_i = \{ j \mid (i, j) \in \mathcal{E} \}$, and its cardinality is given by $d_i = |\mathcal{N}_i|$. Each directed edge $(i, j)$ corresponds to a wireless transmission channel characterized by a channel coefficient that depends on factors such as distance between the nodes, large-scale path loss, small-scale fading, and environmental interference. 

Each wireless link between nodes $i$ and $j$ is modeled by a complex channel coefficient $h_{ij}=|h_{ij}|\angle\phi_{ij}$,
where $|h_{ij}|$ represents the channel gain affected by path loss, shadowing, and fading, and $\phi_{ij}$ represents the propagation phase shift. In this work, we consider only the channel gains, $|h_{ij}|$, to simplify the analysis. In practical wireless systems, the channel phase can be estimated using standard pilot-based channel estimation techniques. Assuming perfect channel state information (CSI), the transmitting nodes perform ideal phase pre-compensation prior to simultaneous transmission. Therefore, the effective channel coefficients reduce to the corresponding real-valued channel magnitudes $|h_{ij}|$, which are adopted throughout this paper for analytical tractability.

Every node must satisfy a maximum average transmit power per coordinate $P_{i}$, which it allocates between transmitting useful model coefficients and injecting artificial noise to guarantee DP. Allocating more power to information transmission improves the received signal strength and accelerates convergence, whereas allocating more power to noise increases privacy protection at the cost of reduced learning accuracy. Designing an appropriate power-splitting strategy is therefore essential to achieve the desired balance between communication reliability, privacy leakage, and model performance. To address these challenges, we develop a scheme composed of two components: an iterative update algorithm that guarantees convergence to a feasible DL solution, and a mechanism for selecting the power fractions to control the privacy-accuracy trade-off.

Let $\mathbf{x}_{i,t}$ denote the local model coefficients of node $i$ at iteration $t$. Before multicasting, node $i$ generates a zero-mean Gaussian noise vector $\boldsymbol{\eta}_{i,t}$ with time-varying variance per coordinate $\sigma_{i,t}^2$. The model coefficients and the noise are then jointly transmitted over the wireless channel by appropriately allocating the available transmit power between these two components. To manage the balance between privacy and learning performance, node $i$ allocates a fraction $\alpha_{i}$ of its available power $P_i$ to transmitting the model coefficients and the remaining fraction $\beta_{i}=1-\alpha_{i}$ to transmit the noise component. Our design philosophy is to control the privacy-accuracy trade-off through the explicit optimization of the power-splitting factors $\alpha_i$ and $\beta_i$. Unlike the channel noise power, which is determined by the wireless environment and is beyond our control, these power allocation factors are design variables that can be adjusted to achieve the desired privacy and learning performance. So, the resulting transmitted signal is constructed as

\begin{equation}\label{eq:transmitted-signal}
    \tilde{\mathbf{x}}_{i,t}
    = \sqrt{\alpha_{i}P_i}\,\mathbf{x}_{i,t}
    + \sqrt{\beta_{i}P_i}\,\boldsymbol{\eta}_{i,t}.
\end{equation}
% This transmitted signal satisfies the following average power constraint:
% \begin{equation*}
% \mathbb{E}\left[
% \left\|
% \sqrt{\alpha_{i} P_i}\,\mathbf{x}_{i,t}
% +
% \sqrt{\beta_{i} P_i}\,\boldsymbol{\eta}_{i,t}
% \right\|^2
% \right]
% \le P_i.
% \end{equation*}

The average transmit power of the transmitted signal can be computed as
\begin{align*}
\mathbb{E}\!\left[\|\tilde{\mathbf{x}}_{i,t}\|^2\right]
&=
P_i\!\left(
\alpha_i\mathbb{E}\!\left[\|\mathbf{x}_{i,t}\|^2\right]
+
\beta_i\mathbb{E}\!\left[\|\boldsymbol{\eta}_{i,t}\|^2\right]
\right),
\end{align*}
where the cross term vanishes because $\mathbf{x}_{i,t}$ and $\boldsymbol{\eta}_{i,t}$ are independent and
$\mathbb{E}[\boldsymbol{\eta}_{i,t}]=\mathbf{0}$. For the transmit power analysis, we additionally require that the model coefficients satisfy $\|\mathbf{x}_{i,t}\|\le\sqrt{m}$, for all \(i\) and \(t\). Moreover, we assume that the injected noise vector is Gaussian with variance $\sigma_{i,t}^2\le1$ per-coordinate, i.e., $\boldsymbol{\eta}_{i,t}\sim
\mathcal N(\mathbf{0},\sigma_{i,t}^2\mathbf I_m)$. Therefore, we have
\[
\mathbb{E}\!\left[\|\mathbf{x}_{i,t}\|^2\right]\le m,
\qquad
\mathbb{E}\!\left[\|\boldsymbol{\eta}_{i,t}\|^2\right]
=m\sigma_{i,t}^2\le m,
\]
from which it follows that
\[
\mathbb{E}\!\left[\|\tilde{\mathbf{x}}_{i,t}\|^2\right]
\le
mP_i(\alpha_i+\beta_i)
=
mP_i,
\]
where $\alpha_i+\beta_i=1$. Consequently,
\[
\frac1m
\mathbb{E}\!\left[\|\tilde{\mathbf{x}}_{i,t}\|^2\right]
\le P_i.
\]

This condition guarantees that the average per-coordinate transmit power does not exceed $P_i$.

Through this mechanism, each neighboring node receives a controlled mixture of the actual model update and DP noise, enabling collaborative learning while maintaining DP guarantees. \Cref{fig:ring-square} illustrates an example DP-DL network with $K=4$ nodes arranged in a ring topology.

%%%%%%%%%%%%%%%%%%%%%%%%%%%%%%%%%%%
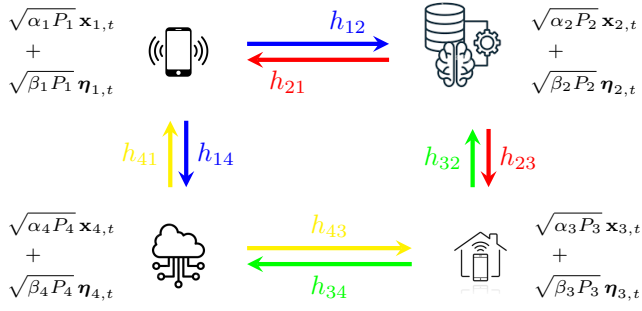
\begin{figure}
\centering
\scalebox{1}{
\begin{tikzpicture}[>=stealth]

\tikzset{
    thick node/.style={minimum size=1.4cm, inner sep=0, outer sep=0}
}

% ============================
% PERFECT SQUARE (4x4)
% ============================
\node[thick node] (1) at (-2, 2) {\includegraphics[width=1cm]{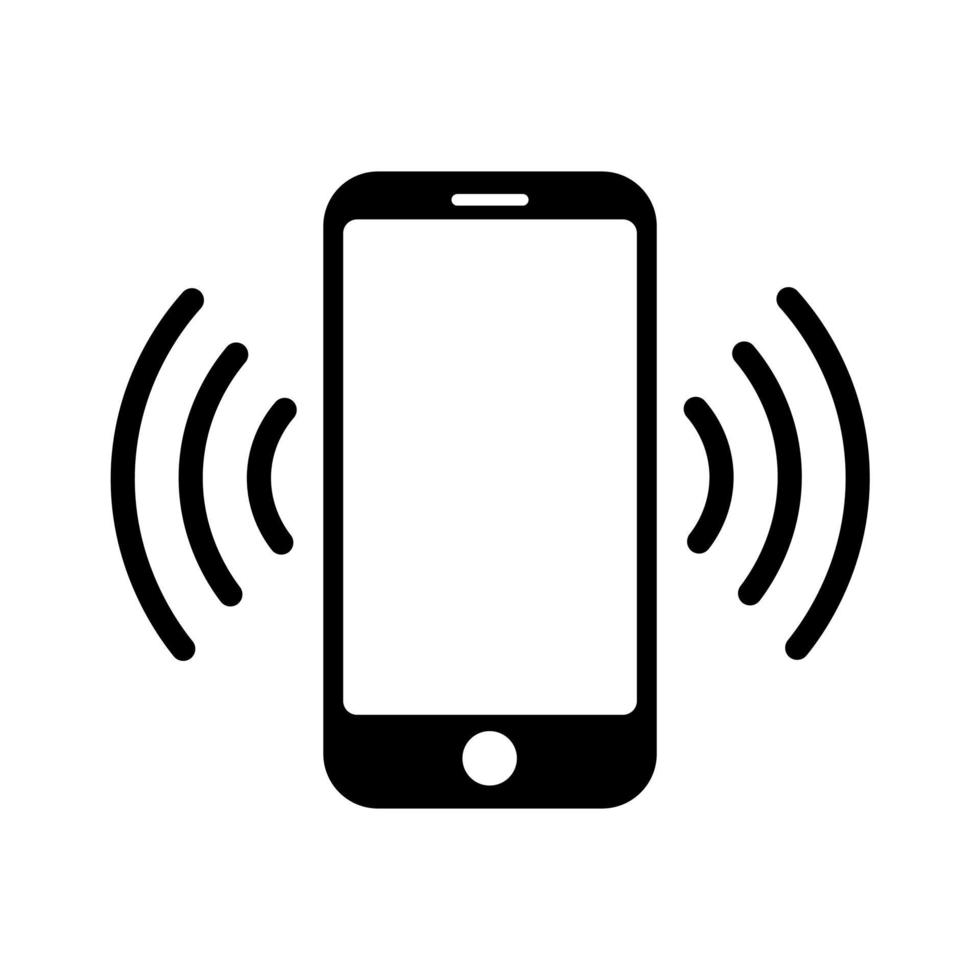}};
\node[thick node] (2) at ( 2, 2) {\includegraphics[width=2cm]{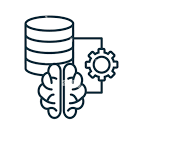}};
\node[thick node] (3) at ( 2,-0.7) {\includegraphics[width=1cm]{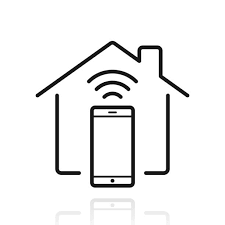}};
\node[thick node] (4) at (-2,-0.7) {\includegraphics[width=1cm]{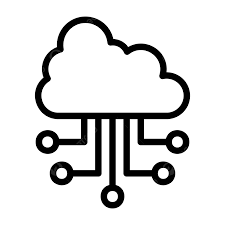}};

% ============================
% LINKS WITH EXTRA SPACING
% ============================

% ---- 1 <-> 2 ----
\draw[->, ultra thick, blue, shorten >=6pt, shorten <=6pt]
    ([yshift=3pt]1.east) -- ([yshift=3pt]2.west)
    node[midway, above right] {$h_{12}$};

\draw[->, ultra thick, red, shorten >=6pt, shorten <=6pt]
    ([yshift=-3pt]2.west) -- ([yshift=-3pt]1.east)
    node[midway, below left] {$h_{21}$};

% ---- 2 <-> 3 ----
\draw[->, ultra thick, red, shorten >=6pt, shorten <=6pt]
    ([xshift=3pt]2.south) -- ([xshift=3pt]3.north)
    node[midway, right] {$h_{23}$};

\draw[->, ultra thick, green, shorten >=6pt, shorten <=6pt]
    ([xshift=-3pt]3.north) -- ([xshift=-3pt]2.south)
    node[midway, left] {$h_{32}$};

% ---- 3 <-> 4 ----
\draw[->, ultra thick, green, shorten >=6pt, shorten <=6pt]
    ([yshift=-3pt]3.west) -- ([yshift=-3pt]4.east)
    node[midway, below] {$h_{34}$};

\draw[->, ultra thick, yellow, shorten >=6pt, shorten <=6pt]
    ([yshift=3pt]4.east) -- ([yshift=3pt]3.west)
    node[midway, above] {$h_{43}$};

% ---- 4 <-> 1 ----
\draw[->, ultra thick, yellow, shorten >=6pt, shorten <=6pt]
    ([xshift=-3pt]4.north) -- ([xshift=-3pt]1.south)
    node[midway, left] {$h_{41}$};

\draw[->, ultra thick, blue, shorten >=6pt, shorten <=6pt]
    ([xshift=3pt]1.south) -- ([xshift=3pt]4.north)
    node[midway, right] {$h_{14}$};

% ============================
% FORMULAS NEXT TO NODES
% ============================

% Node 1 (left)
\node[anchor=east] at ([xshift=0cm]1.west) {\scriptsize$
\begin{aligned}
&\sqrt{\alpha_{1}P_1}\, \mathbf{x}_{1,t} \\
&\quad +\\
&\sqrt{\beta_{1}P_1}\, \boldsymbol{\eta}_{1,t}
\end{aligned}
$};

% Node 2 (right)
\node[anchor=west] at ([xshift=0cm]2.east) {\scriptsize$
\hspace{-15pt}
\begin{aligned} 
&\sqrt{\alpha_{2}P_2}\, \mathbf{x}_{2,t} \\
&\quad +\\
&\sqrt{\beta_{2}P_2}\, \boldsymbol{\eta}_{2,t}
\end{aligned}
$};

% Node 3 (right)
\node[anchor=west] at ([xshift=0cm]3.east) {\scriptsize$
\hspace{-5pt}
\begin{aligned}
&\sqrt{\alpha_{3}P_3}\, \mathbf{x}_{3,t} \\
&\quad +\\
&\sqrt{\beta_{3}P_3}\, \boldsymbol{\eta}_{3,t}
\end{aligned}
$};

% Node 4 (left)
\node[anchor=east] at ([xshift=0cm]4.west) 
{\scriptsize$
\begin{aligned}
&\sqrt{\alpha_{4}P_4}\, \mathbf{x}_{4,t} \\
&\quad +\\
&\sqrt{\beta_{4}P_4}\, \boldsymbol{\eta}_{4,t}
\end{aligned}
$}; 
\end{tikzpicture}
}
\caption{A representative ring topology of the proposed DP-DL framework with $K=4$ clients, where each client splits its power budget $P_i$ between model coefficients $\mathbf{x}_{i,t}$ and Gaussian DP noise $\boldsymbol{\eta}_{i,t}$ using $\alpha_i$ and $\beta_i$, respectively, for transmission over wireless channels $|h_{ij}|$.
}
\vspace{-3ex}
\label{fig:ring-square}
\end{figure}

%%%%%%%%%%%%%%%%%%%%%%%%%%%%%%%%%%%%%%%
At each iteration $t$, all neighbors of node $i$ multicast their constructed signals simultaneously over the shared wireless medium. The channel thus behaves as a Gaussian MAC, and node $i$ receives the superposition of its neighbors’ signals through OAC. The received signal is given by

\begin{equation}\label{eq:received-signal}
    \begin{aligned}
    \mathbf{y}_{i,t}
    &= \sum_{j\in\mathcal{N}_i}
    |h_{ji}|\tilde{\mathbf{x}}_{j,t}\\
    &= \sum_{j\in\mathcal{N}_i}
    |h_{ji}|\Big(
    \sqrt{\alpha_{j}P_j}\mathbf{x}_{j,t}
    +
    \sqrt{\beta_{j}P_j}\boldsymbol{\eta}_{j,t}
    \Big).    
    \end{aligned}
\end{equation}

Additive thermal noise is omitted for simplicity, since the injected DP noise is Gaussian and follows the same statistical form, allowing to treat both sources equivalently. The heterogeneous channel gains $|h_{ji}|$ and transmit power levels $P_j$, together with the power-splitting factors $\alpha_j$ and $\beta_j$, induce different scaling of the model and noise components in the OAC aggregation. The resulting signal is then used for the local model update at node $i$.

\begin{remark}[Effect of Additive Channel Noise]
If additive Gaussian channel noise with variance $\tau_{i,t}^2$ is present, the effective noise variance received at node $i$ becomes
\[
\left(\sigma_{i,t}^{(c)}\right)^2
=
\sum_{k\in\mathcal{N}_i}
|h_{ki}|^2\beta_kP_k\sigma_{k,t}^2
+\tau_{i,t}^2.
\]

Hence, to achieve a target effective noise variance
$\left(\sigma_{i,t}^{\mathrm{req}}\right)^2$, the required aggregate DP noise contribution must satisfy
\[
\sum_{k\in\mathcal{N}_i}
|h_{ki}|^2\beta_kP_k\sigma_{k,t}^2
\geq
\max\!\left\{
0,\,
\left(\sigma_{i,t}^{\mathrm{req}}\right)^2-\tau_{i,t}^2
\right\}.
\]

Thus, the intrinsic channel noise naturally contributes to privacy protection, potentially allowing smaller DP noise power fractions $\beta_k$ and larger model power fractions $\alpha_k$, which can improve learning performance. This provides an additional degree of freedom for balancing privacy and learning performance in wireless networks. 
\end{remark}
% -------------------------------------------------------------
\section{Proposed Algorithm}\label{sec:algorithm}
The proposed method begins by determining how each node allocates its power between transmitting model coefficients and injecting Gaussian noise for DP. Since this allocation directly influences both learning performance and privacy leakage, the first step is to select the largest feasible set of power fractions $\{\alpha_i\}$ that satisfy a prescribed privacy threshold, $\epsilon_{\max}$. This is formulated as the maximization problem
\begin{equation}\label{eq:alpha-LP}
%\begin{aligned}
\max_{0 \le \alpha_j \le 1}\;\sum_{j=1}^K \alpha_j,
\quad \text{s.t.}\;\;\epsilon_{ij} \le \epsilon_{\max}, \quad \forall\, i,j \in \mathcal{V},
%\end{aligned}
\end{equation}
where $\epsilon_{ij}$ is the privacy leakage experienced by node $i$ due to the transmission of node $j$. A detailed derivation of the expression for $\epsilon_{ij}$ is given in \cref{thm:privacy} in \cref{sec:theoretical-analysis}. Although the privacy leakage expression in Theorem~1 depends on the iteration index $t$ through the learning rate $\gamma_t$ and the DP noise variance $\sigma_{i,t}^2$, the proposed framework adopts a time-varying DP noise standard deviation proportional to the learning rate, i.e., $\sigma_{i,t} \propto \gamma_t$. As a result, the time-dependent terms cancel in the privacy leakage expression, making $\epsilon_{ij}$ independent of $t$. Consequently, the optimization problem reduces to a static linear program whose solution is computed once before the training procedure and remains fixed throughout all iterations. The resulting $\alpha_i$ and $\beta_i = 1 - \alpha_i$ remain fixed for the entire learning process. The constraint $\epsilon_{ij} \le \epsilon_{\max}$ ensures that, for every directed communication link $(j \rightarrow i)$, the privacy leakage resulting from node $j$'s transmitted signal remains below a predefined tolerance level. 

The quantity $\epsilon_{ij}$ captures how much the received signal at node $i$ can change when a single data sample in node $j$’s dataset is modified; hence, it represents the local DP risk imposed by $j$ on $i$. By enforcing that each leakage term stays below $\epsilon_{\max}$, the system provides a uniform privacy guarantee across the entire network. A larger $\alpha_i$ yields stronger received model signals and directly improves convergence speed and robustness; however, any increase in $\alpha_i$ relaxes the DP constraint, potentially increasing the privacy leakage $\epsilon_{ij}$. The maximization problem therefore identifies the largest feasible set of power-splitting coefficients that simultaneously improve learning performance and satisfies the DP budget under a given threshold $\epsilon_{\max}$. These optimized power fractions are then utilized in the update rule and adjacency matrix construction, forming the basis of the proposed DP-DL algorithm.

Once the optimal power-splitting coefficients $\alpha_i$ and $\beta_i$ are obtained, each node $i$ constructs its privacy-protected signal $\tilde{\mathbf{x}}_{i,t}$ using \cref{eq:transmitted-signal}. 
This signal embeds both the model coefficients and DP noise. Then, each node $i$ multicasts its constructed signal through the wireless channel to all of its neighbors. Due to the superposition property of wireless communication and the use of OAC, each receiving node obtains an 
aggregated signal that combines the contributions from its neighbors. 
Node $i$ receives the OAC-aggregated signal $\mathbf{y}_{i,t}$, whose expression is provided in~\cref{eq:received-signal}. Accordingly, the local update rule should be designed to simultaneously achieve three objectives: efficiently utilize the aggregated wireless signal, preserve privacy through injected DP noise, and maintain convergence in the presence of the resulting noise perturbations.
Motivated by the approaches in~\cite{9013030} and \cite{10025677} and given the received signal $\mathbf{y}_{i,t}$, each node updates its local model coefficients by using the modified DSGD update rule in \cref{eq:step-local-update} as follows:
\begin{equation}\label{eq:node-update-rule}
\begin{aligned}
    &\mathbf{x}_{i,t+1} 
    \\&= 
    %\raisebox{0pt}{\scalebox{1.5}{$\Pi$}}_{\Omega} 
    \Pi_{\Omega}
    \Bigg (\frac{\mathbf{y}_{i,t}}{c_{i} (d_{i}+1}) + \Big(\frac{1} {d_{i}+1}\Big) \Big(\mathbf{x}_{i,t} + \sqrt{\frac{\beta_{i}}{\alpha_{i}}} \boldsymbol{\eta}_{i,t}\Big) -\gamma_{t}\frac{\mathbf{g}_{i,t}}{z_{ii,t}}\Bigg)\\
    &= 
    %\raisebox{0pt}{\scalebox{1.5}{$\Pi$}}_{\Omega}
    \Pi_{\Omega}
    \Bigg(\frac{\sum_{j \in \mathcal{N}_{i}} (|h_{ji}|\sqrt{\alpha_{j} P_{j}}\mathbf{x}_{j,t}+|h_{ji}|\sqrt{\beta_{j} P_{j}}\boldsymbol{\eta}_{j,t})}{c_{i} (d_{i}+1)} \\
    &+ \Big(\frac{1} {d_{i}+1}\Big) \Big(\mathbf{x}_{i,t} + \sqrt{\frac{\beta_{i}}{\alpha_{i}}} \boldsymbol{\eta}_{i,t}\Big)- \gamma_{t}\frac{\mathbf{g}_{i,t}}{z_{ii,t}}\Bigg),
\end{aligned}    
\end{equation}
where $\gamma_t$ denotes the time-decreasing learning rate and $\mathbf{g}_{i,t}$ represents the local stochastic gradient of the objective function at node $i$ evaluated at $\mathbf{x}_{i,t}$. 

Additionally, the term $z_{ii,t}$ represents the $i$-th component of the auxiliary vector $\mathbf{z}_{i,t}$, which is exchanged together with $\mathbf{x}_{i,t}$ between the neighboring nodes. As shown in~\cite{MAI201994}, we have $z_{ii,t}>0$ converging to $\pi_i$, the $i$-th entry of the left Perron eigenvector of the adjacency matrix. Since $z_{ii,t}>0$ for all $i$ and $t$, there exists a finite constant $\theta$ such that 
\begin{equation}
\label{eq:z_bound}
\frac{1}{z_{ii,t}} \le \theta, \quad \forall i \in \mathcal{V}, \; \forall t \ge 0.
\end{equation}

Since $\mathbf{z}_{i,t}$ contains only $K$ scalar entries, its $O(K)$ communication overhead is relatively small compared with the cost of transmitting high-dimensional model updates, particularly for the network sizes considered in this paper. Moreover, because $\mathbf{z}_{i,t}$ does not depend on the underlying raw
data, its transmission does not introduce any privacy leakage. Furthermore, to guarantee convergence of the iterative update, the compensation term $c_{i}$ is chosen as
\begin{equation}\label{eq:c_i}
     c_{i}
    = \frac{1}{d_i}
      \sum_{j \in \mathcal{N}_{i}}
      |h_{ji}| \sqrt{\alpha_{j} P_j}.
\end{equation}

With this choice of $c_i$, the effective coefficients assigned to the neighboring models sum to $d_i/(d_i+1)$, while the self-weight is $1/(d_i+1)$, ensuring that the resulting mixing weights sum to one. The normalization factor $c_i$ is computed locally at each node. Specifically, node $i$ estimates the channel magnitudes ${|h_{ji}|}_{j\in\mathcal{N}_i}$ through standard pilot-based channel estimation. Therefore, since the transmit powers ${P_j}$ and scaling coefficients ${\alpha_j}$ are predetermined offline-configured system parameters (or exchanged during an initialization phase), node $i$ can compute its own $c_i$ using only locally available information. So, the proposed framework does not require any centralized computation or additional coordination during the iterative learning process.

Using \cref{eq:c_i} for calculating $c_{i}$, the update procedure in \cref{eq:node-update-rule} can be rewritten in the more compact form
\begin{equation}\label{eq:update-procedure}
\begin{aligned}
\mathbf{x}_{i,t+1} = 
%\raisebox{0pt}{\scalebox{1.5}{$\Pi$}}_{\Omega}
\Pi_{\Omega}
\Big(\sum_{j=1}^{K} a_{ij} \big(\mathbf{x}_{j,t} + \sqrt{\tfrac{\beta_j}{\alpha_j}}\boldsymbol{\eta}_{j,t}\big) - \gamma_t\frac{\mathbf{g}_{i,t}}{z_{ii,t}} \Big),
\end{aligned}
\end{equation}
where the weights
\begin{equation*}
a_{ij} = \frac{|h_{ji}|\sqrt{\alpha_j P_j}}{c_i (d_{i}+1)} \text{ for }i \neq j, \text{ and }
a_{ii} = \frac{1}{d_{i}+1},
\end{equation*}
are exactly those that make $A$ row-stochastic, meaning that $\sum_{j=1}^{K} a_{ij} = 1 \quad, \forall  i \in \mathcal{V}$.

More precisely, node $i$ computes $\mathbf{x}_{i,t+1}$ using a combination of three terms: 
(i) the OAC-aggregated neighbor information, 
(ii) its own noise-perturbed model coefficients, and 
(iii) a gradient descent adjustment scaled by the auxiliary variable~$z_{ii,t}$. 

The auxiliary variables follow the linear update rule
\begin{equation}\label{eq:z_i-update-rule}
\mathbf{z}_{i,t+1} = \sum_{j=1}^K a_{ij} \mathbf{z}_{j,t},    
\end{equation}
which is initialized as $\mathbf{z}_{i,0} = \mathbf{e}_i$, the $i$-th standard unit vector, i.e., a one-hot vector with its $i$-th entry equal to $1$. 
Under repeated linear updates, the $i$-th component $z_{ii,t}$ converges to $\pi_i$, the corresponding entry of the left Perron eigenvector. 

The algorithm operates in an iterative manner until convergence, and its overall procedure is summarized in \cref{alg:main}.

To analyze the behavior of the decentralized update rule and establish convergence guarantees, we introduce several standard assumptions. These assumptions specify the properties of the local objective functions, the constraint set, and the gradient bounds used in the update steps. 
\begin{assumption}\label{ass:strong-convexity}
Each local cost function $f_i$,
is \( \mu \)-strongly convex. That is, for any \( x, y \in \Omega \),
\[
f_i(y) \ge f_i(x) + \nabla f_i(x)^\top (y - x) 
+ \frac{\mu}{2}\|y - x\|^2,
\]
where \( \mu > 0 \) denotes the strong convexity constant. 
\end{assumption}

Strong convexity ensures that each \( f_i \) has a unique minimizer and provides curvature that leads to faster convergence of gradient-based 
algorithms. While the theoretical analysis assumes strong convexity, the experiments demonstrate empirical robustness of the algorithm beyond this regime (see \cref{sec:sim-results}), consistent with common practice in DL literature.

\begin{assumption}\label{ass:feasible-set}
The constraint set $\Omega$ is nonempty, convex, and closed, and $\mathbf{0}\in\Omega$. Furthermore, its diameter is bounded by a constant $L<\infty$, such that
$\|\mathbf{x}-\mathbf{y}\|\le L, \forall\,\mathbf{x},\mathbf{y}\in\Omega$.
Throughout this paper, we assume that
$L\le\sqrt{m}$.
\end{assumption}
Convexity guarantees that the constraint region contains all convex 
combinations of its points, closedness ensures the existence of 
solutions within \( \Omega \), and boundedness prevents iterates from 
escaping to infinity.

\begin{assumption}\label{ass:bounded-gradients}
For every node $i$, every iteration \( t \ge 0 \), 
and any \( x \in \Omega \), the stochastic gradient of the local loss function $\mathbf{g}_{i,t}$,
is uniformly bounded. Specifically, there exists a constant \( G > 0 \) 
such that
$\|\mathbf{g}_{i,t}\| \le G$.
\end{assumption}

This assumption prevents gradient magnitudes from exploding, ensures 
stability in the update process, and is essential for deriving 
convergence and regret bounds in DL methods. 

To evaluate the learning performance over a finite horizon of $T$ iterations, we use the following cumulative regret function \cite{6930789}:
\begin{equation}\label{eq:regret-fun}
\begin{aligned}
\mathbb{R}_{j}(T)
&\triangleq
\sum_{t=1}^{T}
\left(
F(\mathbf{x}_{j,t})
-
F(\mathbf{x}^{*})
\right) \\
&=
\sum_{t=1}^{T}
\left(
\sum_{i=1}^{K}
f_i(\mathbf{x}_{j,t})
-
\sum_{i=1}^{K}
f_i(\mathbf{x}^{*})
\right).
\end{aligned}
\end{equation}

This  measures how much worse node $j$ performs compared to the optimal solution $\mathbf{x}^*$ that has full knowledge of all datasets. Although problem \cref{eq:optim-problem} is formulated as a static offline optimization problem, the iterative decentralized updates involving stochastic gradients, exchanged model coefficients, and time-varying DP perturbations inherently generate a sequential learning process across iterations. Therefore, although regret is commonly used in online optimization, it can also be used to analyze the optimization process of iterative offline algorithms. Accordingly, the cumulative regret in~\cref{eq:regret-fun} quantifies the accumulated optimization error and characterizes convergence over the training horizon, rather than modeling an online optimization problem.

%%%%%%%%%%%%%%%%%%%%
\begin{algorithm}[t]
\caption{Joint Power-Privacy Controlled DP-DL}
\begin{algorithmic}[1] % The '[1]' enables line numbers
\State \textbf{Inputs:} Constraint set $\Omega$, number of iterations $T$, the power budgets of the nodes $P_{i}$, the prescribed maximum privacy threshold $\epsilon_{\max}$, and learning rate $\{\gamma_t\}_{t=1}^T$ 
\State \textbf{Initialization:} Initialize $\mathbf{x}_{i,0} \in \Omega$ and set $\mathbf{z}_{i,0} = \mathbf{e}_i, \forall i \in \mathcal{V}$ and solve the maximization problem \cref{eq:alpha-LP} to find optimal $\alpha_{i}$ and $\beta_{i}$
\For{each iteration $t = 0, 1, 2, \ldots, T-1$}
    \For{each node $i \in \mathcal{V}$ in parallel}
        \State Generate noise $\boldsymbol{\eta}_{i,t} \sim \mathcal{N}(\mathbf0,\sigma_{i,t}^2\mathbf I_m)$
        \State Construct $\tilde{\mathbf{x}}_{i,t}=\sqrt{\alpha_{i} P_{i}}\mathbf{x}_{i,t} + \sqrt{\beta_{i} P_{i}}\boldsymbol{\eta}_{i,t}$
        \State Multicast  $\tilde{\mathbf{x}}_{i,t}$ and $\mathbf{z}_{i,t}$ to all neighbors $j \in \mathcal{N}_i$
        \Statex \hspace{1cm} over the channels $h_{ij}$ and compute the resulting
        \Statex \hspace{1cm} aggregated $\mathbf{y}_{i,t}$ using \cref{eq:received-signal}  
        \State Update $\mathbf{x}_{i,t+1}$ using \cref{eq:node-update-rule}
        \State Update $\mathbf{z}_{i,t+1}$ using \cref{eq:z_i-update-rule}
    \EndFor
\EndFor
\State \textbf{Outputs:} Model coefficients $\{\mathbf{x}_{i,t}\}_{i \in \mathcal{V}}$ for $t=1,2,...,T$
\end{algorithmic}
\label{alg:main}
\end{algorithm}
%%%%%%%%%%%%%%%%%%%%
% -------------------------------------------------------------
\section{Theoretical Analysis}\label{sec:theoretical-analysis}

% -----------------------------------------------------------

In this section, we establish the privacy guarantees and analyze the convergence of the proposed algorithm under Gaussian DP noise and heterogeneous network settings.
\subsection{Privacy}
Our first major result, stating the $(\epsilon_{ij},\delta)$-DP for each node is given as follows. 
\begin{theorem}\label{thm:privacy}
The proposed algorithm achieves $(\epsilon_{ij}, \delta)$-DP for node $i$ with respect to its neighbor $j$ at each iteration $t$, where $\epsilon_{ij}$ is given by
\begin{equation}\label{eq:e_ij-thm}
\begin{aligned}
    \epsilon_{ij}
    % = \frac{ \Delta_{ij}^{t}}{\sigma_{i,t}^{r}}\sqrt{2 \ln \frac{1.25}{\delta}}\\
    =\frac{2G \gamma_t \theta |h_{ji}|\sqrt{\alpha_j P_j}}
    {\sqrt{\sum_{k \in \mathcal{N}_i} |h_{ki}|^2 \beta_k P_k \sigma_{k,t}^2}}
    \sqrt{2 \ln \frac{1.25}{\delta}}.
\end{aligned}
\end{equation}   
% where $\sigma_{i,t}^{r}$ denotes the standard deviation of the effective noise received at node $i$ during iteration $t$.
\end{theorem}
\begin{IEEEproof}
    To quantify the DP guarantee of the proposed algorithm, we analyze how much the received signal at node $i$ can change when a single data point in the dataset of node $j$ is modified. Considering two neighboring datasets $\mathcal{D}_j$ and $\mathcal{D}_j'$ at node $j$, and applying the sensitivity definition in \cref{eq:sensitivity-def} together with \cref{eq:received-signal} and \cref{eq:node-update-rule}, the resulting sensitivity can be upper-bounded as
\begin{equation}\label{eq:sensitivity-UB}
\begin{aligned}
    \Delta_{ij}^{t}&= \max_{\mathcal{D}_{j},\mathcal{D}_j'} \|\mathbf{y}_{i,t}(\mathcal{D}_{j})-\mathbf{y}_{i,t}(\mathcal{D}_j')\|\\
    &= \max_{\mathcal{D}_{j},\mathcal{D}_j'} \||h_{ji}|\sqrt{\alpha_{j}P_{j}} (\mathbf{x}_{j,t}(\mathcal{D}_{j})-\mathbf{x}_{j,t}(\mathcal{D}_j'))\|\\
    &= \max_{\mathcal{D}_{j},\mathcal{D}_j'} \||h_{ji}|\sqrt{\alpha_{j}P_{j}} \frac{\gamma_{t}}{z_{jj,t}} (\mathbf{g}_{j,t}(\mathcal{D}_{j})-\mathbf{g}_{j,t}(\mathcal{D}_j'))\|\\
    &\leq 2G\gamma_{t}\theta |h_{ji}|\sqrt{\alpha_{j}P_{j}}.
\end{aligned}    
\end{equation}
% where $G$ bounds all gradient norms and $\theta$ upper bounds $1/z_{ii,t}$, $\forall i \in \mathcal{V}$ and all $t \ge 0$. 

On the other hand, the noise observed at receiver $i$ at time $t$ results from the superposition of independently injected Gaussian noise signals transmitted by its neighboring nodes over the wireless channels. Specifically, each neighbor $k \in \mathcal{N}_i$ adds a zero-mean Gaussian noise vector $\boldsymbol{\eta}_{k,t} \sim \mathcal{N}(\mathbf{0}, \sigma_{k,t}^2 \mathbf{I})$ to its local model coefficients. During transmission, this noise is scaled by $\sqrt{\beta_kP_k}$ and the channel magnitude $|h_{ki}|$. The aggregate noise masks each neighbor's contribution, thereby providing privacy at node $i$. Accordingly, the standard deviation of the combined noise at node $i$ is given by
\begin{equation*}
 \sigma_{i,t}^{(c)}=\sqrt{\sum_{k \in \mathcal{N}_i} |h_{ki}|^2 \beta_k P_k \sigma_{k,t}^2}.   
\end{equation*}

By substituting $\Delta_{ij}^t$ and $\sigma_{i,t}^{(c)}$ into the Gaussian mechanism in \cref{eq:sigma-DP}, we obtain the statement of the theorem.

\end{IEEEproof}

To facilitate an efficient solution, we show that \eqref{eq:alpha-LP} can be reformulated as an LP problem. Although the privacy constraints are initially nonlinear due to the square-root ratios, they become linear inequalities under the parameter setting used in this framework.

From~\cref{thm:privacy}, the privacy constraint is given by
\[
\frac{
2G\gamma_t\theta |h_{ji}|\sqrt{\alpha_jP_j}
}{
\sqrt{
\sum_{k\in\mathcal{N}_i}
|h_{ki}|^2\beta_kP_k\sigma_{k,t}^{2}
}
}
\sqrt{
2\ln\left(\frac{1.25}{\delta}\right)
}
\leq
\epsilon_{\max}.
\]

Assuming $\sigma_{k,t}=\rho_k\gamma_t$,
where $\rho_k>0$ is a constant, the common factor $\gamma_t$ appearing in both the numerator and denominator cancels. Finally, substituting $\beta_k=1-\alpha_k$ yields
\[
\begin{aligned}
&
8G^2\theta^2|h_{ji}|^2P_j
\ln\!\left(\frac{1.25}{\delta}\right)\alpha_j
+
\epsilon_{\max}^2
\sum_{k\in\mathcal{N}_i}
|h_{ki}|^2P_k\rho_k^2\alpha_k
\\
&\hspace{5em}\leq
\epsilon_{\max}^2
\sum_{k\in\mathcal{N}_i}
|h_{ki}|^2P_k\rho_k^2,
\end{aligned}
\]
which is a linear inequality in the optimization variables $\{\alpha_k\}$. Since the objective function is also linear, the optimization problem in \eqref{eq:alpha-LP} is a linear programming problem.

%%%%%%%%%%%%%%%%%%%%%%%%%%%%%%%%%%%%%%%%%%%
\subsection{Privacy Composition and Accounting}
Theorem~\ref{thm:privacy} provides a per-iteration local DP guarantee. Although cumulative leakage can be bounded using standard or advanced composition~\cite{dwork2014algorithmic}, tighter accounting can be obtained using R\'enyi differential privacy (RDP) with privacy amplification by Poisson subsampling. If each data point is sampled independently with probability $q$, an $(\epsilon_{ij},\delta)$-DP mechanism becomes $\left(\log\left(1+q(e^{\epsilon_{ij}}-1)\right),q\delta\right)$-DP~\cite{Balle2018PrivacyAmplification}. We therefore use the corresponding subsampled RDP accountant, denoting its order-$\lambda$ parameter by $\epsilon_{ij}^{\mathrm{sub}}(\lambda)$. According to Theorem~9 in~\cite{DBLP:journals/corr/abs-1808-00087},
\[
\begin{aligned}
\epsilon_{ij}^{\mathrm{sub}}(\lambda)
&\leq \frac{1}{\lambda-1}
\ln\Bigg[
1+q^2\binom{\lambda}{2}
\min\!\left\{4\!\left(e^{\epsilon_{ij}(2)}-1\right),
2e^{\epsilon_{ij}(2)}\right\}
\\
&+
\sum_{k=3}^{\lambda}
2q^k\binom{\lambda}{k}
\exp\!\left((k-1)\epsilon_{ij}(k)\right)
\Bigg],
\end{aligned}
\]
where $\epsilon_{ij}(k) =\frac{k}{2}\left(\frac{\Delta_{ij}^{t}}{\sigma_{i,t}^{(c)}} \right)^2$
denotes the RDP parameter of the original (non-subsampled) Gaussian mechanism at order $k$.

By the additive composition property of RDP~\cite{DBLP:journals/corr/Mironov17}, the cumulative privacy leakage after $T$ communication rounds is
$\epsilon_{ij}^{\mathrm{tot}}(\lambda)
=
\sum_{t=1}^{T}
\epsilon_{ij,t}^{\mathrm{sub}}(\lambda)$.
Since the ratio
$\Delta_{ij}^{t}/\sigma_{i,t}^{(c)}$
is independent of the communication round, the per-iteration RDP parameter is also time-invariant, i.e.,
$\epsilon_{ij,t}^{\mathrm{sub}}(\lambda)=\epsilon_{ij}^{\mathrm{sub}}(\lambda)$.
Therefore,
$\epsilon_{ij}^{\mathrm{tot}}(\lambda)
=
T\,\epsilon_{ij}^{\mathrm{sub}}(\lambda)$.
Finally, the corresponding $(\epsilon,\bar{\delta})$-DP guarantee is obtained by
\[
\epsilon_{ij,\mathrm{sub}}^{\mathrm{RDP}}
=
\min_{\lambda>1}
\left\{
\epsilon_{ij}^{\mathrm{tot}}(\lambda)
+
\frac{\ln(1/\bar{\delta})}{\lambda-1}
\right\},
\]
where $\bar{\delta}>0$ is the target failure probability.

% ===========
\subsection{Convergence Rate}
We use the regret function in \cref{eq:regret-fun} to characterize cumulative error and convergence. The analysis addresses two main challenges: DP noise perturbations and row-stochastic mixing caused by asymmetric channel gains. We first introduce the following lemma:
\begin{lemma}[\!\cite{MAI201994}]\label{lem:convergence-to-pi}
Let $A$ be a given row-stochastic adjacency matrix associated with a fixed and
strongly connected communication graph $\mathcal{G}$. Let $\boldsymbol{\pi}$ be
the unique left Perron eigenvector of $A$, satisfying 
$\boldsymbol{\pi}^{\top} A = \boldsymbol{\pi}^{\top}$, $\boldsymbol{\pi}^{\top} \mathbf{1}_{n} = 1$ and $\sum_{i} \pi_{i} = 1$.  
Under \cref{ass:strong-convexity,ass:feasible-set}, there exist constants $C>0$ and $0<\xi<1$ such that
for all $i,j \in \mathcal{V}$ and all $t \ge 0$,
\begin{equation*}\label{eq:convergence-to-pi}
    \big| [A^{t}]_{ij} - \pi_{j} \big| \le C \, \xi^{t},
    \qquad
    \big| z_{ii,t} - \pi_{i} \big| \le C \, \xi^{t}.
\end{equation*}
\end{lemma}

Building upon this lemma, the main convergence theorem is now established as follows:
\begin{theorem}\label{Theorem2_convergence}
Under \cref{ass:strong-convexity,ass:feasible-set,ass:bounded-gradients} and using \cref{lem:convergence-to-pi}, the proposed algorithm achieves an expected regret of order $O(\log T)$.
In particular, by choosing a diminishing step size $\gamma_t = 1/(\mu \theta t)$ with constants $\mu > 0$ and $\theta > 0$, the expected regret of the algorithm can be upper-bounded as follows:
\begin{equation*}
\begin{aligned}
    \mathbb{E}[\mathbb{R}_{j}(T)] \leq U_{1} + U_{2}(1+\log T), 
\end{aligned}   
\end{equation*}
where $U_1$ and $U_2$ are defined as
\begin{equation*}
\begin{aligned}
    &U_{1} = \frac{\xi C G}{1-\xi} \Big (2(K+\theta) \sum_{i=1}^{K}\|x_{i,0}\| + K\theta L \Big ), \\
    &U_{2} = O\left(K^{3} m G^{2} \ln \frac{1.25}{\delta}\right)
    % &U_{2} =  \frac{2 S (8K^{2}+8\theta K + \theta )\sqrt{m} G^{2}  \sqrt{2 \ln \frac{1.25}{\delta}} }{\min_{i,j,k} \Big( \epsilon_{ij} \sqrt{\sum_{k \in \mathcal{N}_i} |h_{ki}|^2 \beta_k P_k }\Big)} \\
    % &+ \frac{4 \mu S^{2} m G^{2} \theta \ln (\frac{1.25}{\delta})}{\min_{i,j,k} \Big( \epsilon_{ij} \sqrt{\sum_{k \in \mathcal{N}_i} |h_{ki}|^2 \beta_k P_k }\Big)} \\ 
    % &+ \frac{4 S (2K+1)(\theta + K) K \sqrt{m} C G^{2} \sqrt{2 \ln \frac{1.25}{\delta}}}{(1-\xi)\min_{i,j,k} \Big( \epsilon_{ij} \sqrt{\sum_{k \in \mathcal{N}_i} |h_{ki}|^2 \beta_k P_k }\Big)}\\
    % &+ \frac{2K (K + \theta) C G^{2}}{1-\xi} + \frac{1}{2}(8K+9\theta)G^{2}
\end{aligned}   
\end{equation*}
\end{theorem}
The constants in Theorem~\ref{Theorem2_convergence} are based on worst-case system parameters and are therefore conservative; however, this affects only the constants and not the $O(\log T)$ regret rate. The big O-notation in $U_2$ captures the cumulative contribution of injected DP noise and network interactions to the regret growth over time. In particular, $U_2$ aggregates the effects of (i) multi-node coupling through the mixing process, which introduces polynomial dependence on $K$, (ii) model dimensionality $m$ through the variance of the gradient and noise terms, and (iii) the privacy requirement, which enters logarithmically via $\ln(1.25/\delta)$ through the Gaussian mechanism. The logarithmic regret rate follows from the strongly convex setting together with the diminishing step size $\gamma_t$, which is standard in decentralized learning. Specifically, the logarithmic dependence on $T$ arises from the cumulative effect of the diminishing learning rate in the regret analysis. The role of the proposed privacy-power mechanism is reflected in the constants of the regret bound rather than in changing the asymptotic order. In particular, the terms involving the injected DP noise, the power-splitting coefficients $\{\alpha_i,\beta_i\}$, the transmit powers $\{P_i\}$, and the wireless channel gains $\{|h_{ij}|\}$ enter the regret bound through the noise-dependent components of $U_2$.

The complete proof, along with the exact expression for $U_2$, is provided in \cref{app:proof}. The proof borrows some techniques from \cite{9013030}.

From the definition of $\mathbb{R}_j(T)$ in \eqref{eq:regret-fun}, the regret bound in Theorem~\ref{Theorem2_convergence} implies
$$
\frac{\mathbb{E}[\mathbb{R}_{j}(T)]}{T}
=
\frac{1}{T}
\mathbb{E}\!\left[
\sum_{t=1}^{T}
\left(
F(\mathbf{x}_{j,t})-F(\mathbf{x}^{*})
\right)
\right]
=
O\!\left(\frac{\log T}{T}\right),
$$
which converges to zero as $T\to\infty$. Moreover, since $F(\mathbf{x})=\sum_{i=1}^{K}f_i(\mathbf{x})$, this can be written as
$$
\frac{1}{T}
\sum_{t=1}^{T}
\mathbb{E}\!\left[
\sum_{i=1}^{K}
\left(
f_i(\mathbf{x}_{j,t})-f_i(\mathbf{x}^{*})
\right)
\right]
=
O\!\left(\frac{\log T}{T}\right).
$$

Therefore, the proposed algorithm achieves vanishing average regret, establishing convergence in the time-averaged objective (ergodic) sense.

\begin{Corollary}
Under the assumptions of~\cref{Theorem2_convergence},
\[
\mathbb{E}\!\left[
\min_{1\le t\le T}
\left\{
F(\mathbf{x}_{j,t})-F(\mathbf{x}^{*})
\right\}
\right]
=
O\!\left(\frac{\log T}{T}\right).
\]
\end{Corollary}

\begin{IEEEproof}
Since the minimum of a finite set of nonnegative numbers does not exceed
their average,
\[
\min_{1\le t\le T}
\left\{
F(\mathbf{x}_{j,t})-F(\mathbf{x}^{*})
\right\}
\le
\frac{\mathbb{R}_j(T)}{T}.
\]

Taking expectations and applying~\cref{Theorem2_convergence} yields
\[
\mathbb{E}\!\left[
\min_{1\le t\le T}
\left\{
F(\mathbf{x}_{j,t})-F(\mathbf{x}^{*})
\right\}
\right]
\le
\frac{\mathbb{E}[\mathbb{R}_j(T)]}{T}
=
O\!\left(\frac{\log T}{T}\right),
\]
which completes the proof.
\end{IEEEproof}

% -------------------------------------------------------------
\section{Simulation Results}\label{sec:sim-results}
We conduct a set of simulation experiments to thoroughly evaluate our algorithm. The analysis explores the effects of heterogeneous channel gains, injected DP noise, IID and non-IID data distributions, the number of clients, and communication topology, along with comparisons with existing methods.
\subsection{Simulation Setup}
We evaluate the proposed algorithm on the CIFAR-10 image-classification task, which contains 60{,}000 color images of size $32\times 32$ across ten classes. 
Each client employs a ResNet\text{-}20 architecture, comprising convolutional layers arranged within residual blocks. Its compact size and efficient layer composition make it well suited for DL configurations. In all experiments, $f_i(\mathbf{x})$ denotes the empirical cross-entropy loss of the ResNet-20 model over client $i$'s local dataset, and $F(\mathbf{x})=\sum_{i=1}^{K}f_i(\mathbf{x})$ represents the aggregate loss across all clients. While convexity is assumed for tractable theoretical analysis, the experiments employ nonconvex neural networks and demonstrate the effectiveness of the proposed algorithm beyond the theoretical setting.

%All simulations are implemented on a GPU server to ensure consistency and reproducibility.
We evaluate the scalability of the proposed DP-DL algorithm for $K=4$, $10$, and $20$ clients. The considered topologies are fully connected and random Erd\H{o}s--R\'enyi graphs with $p=0.4$. For $K=4$, the random topology is replaced by a ring topology, since it frequently fails to satisfy the required strong connectivity condition. The dataset is split into 80\% training and 20\% testing data, with the training set distributed among clients under IID and non-IID settings and a common test set used for evaluation. IID data are uniformly distributed among clients, while non-IID data are generated using
$\mathbf{p}_{i}\sim\operatorname{Dirichlet}(\hat{\alpha}\mathbf{1})$,
where $\mathbf{p}_{i}=(p_{i1},\ldots,p_{iC})$ represents the class distribution at client $i$, and $C$ is the number of classes. Smaller $\hat{\alpha}$ produces more skewed class distributions, while larger values approach the IID setting. We set $\hat{\alpha}=1$, which represents a moderate level of non-IID partitioning. This choice avoids extremely skewed client datasets while still capturing realistic heterogeneity across clients. Finally, we examine the effect of DP noise by considering 
maximum privacy budgets $\epsilon_{\max} \in \{0.5,\, 1,\, \infty\}$.
The first two settings introduce nonzero DP noise, whereas $\epsilon_{\max} = \infty$ corresponds to a baseline without any privacy noise. 
% This allows us to quantify the trade-off between privacy protection and learning performance under decentralized training. 
% \textcolor{blue}{For notational simplicity, the privacy threshold $\epsilon_{\max}$ is denoted by $\epsilon$ in all simulation figures.}

%%%%%%%%%%%%%%%%%%%%

\subsection{Results and Discussion} 
To assess the effectiveness of the proposed framework, we compare it with the PED$^2$FL algorithm~\cite{10025677}. The original PED$^2$FL framework assumes a doubly stochastic adjacency matrix and identical wireless channel gains from each transmitting node to all of its neighbors, i.e., $h_{ij}=h_i$ for all $j\in\mathcal{N}_i$. Therefore, to ensure a fair comparison, we adopted the same channel coefficient setting as in the PED$^2$FL framework. We first computed the compensation factor $c$ according to PED$^2$FL, followed by the corresponding power allocation factors $\alpha_i$ and $\beta_i$. Using these parameters, we evaluated the privacy leakage of PED$^2$FL and obtained the individual per-iteration privacy budgets $\epsilon_i$, whose maximum value was denoted by $\epsilon_{\max}$. We then used this $\epsilon_{\max}$ as the privacy upper bound in our proposed optimization problem to determine the corresponding optimal power allocation factors $\alpha_i$ and $\beta_i$. Finally, these optimized parameters were employed to train our proposed algorithm and evaluate its average test accuracy, enabling a fair comparison under the same maximum privacy leakage. Figures~\ref{fig:Comparison with PED2FL with epsilon = 0.5} and~\ref{fig:Comparison with PED2FL_epsilon half} compare the proposed algorithm, Alg.~\ref{alg:main}, with PED$^2$FL in a fully connected 4-client network under the same maximum privacy leakage constraints. For both $\epsilon_{\max}=1$ and $\epsilon_{\max}=0.5$, the proposed method consistently achieves higher test accuracy throughout training under both IID and non-IID data distributions.

% ============================
% Figure  — Comparison with PED2FL with epsilon = 1
% ============================
\begin{figure}[t]
    \centering
    \includegraphics[width=0.9\linewidth]{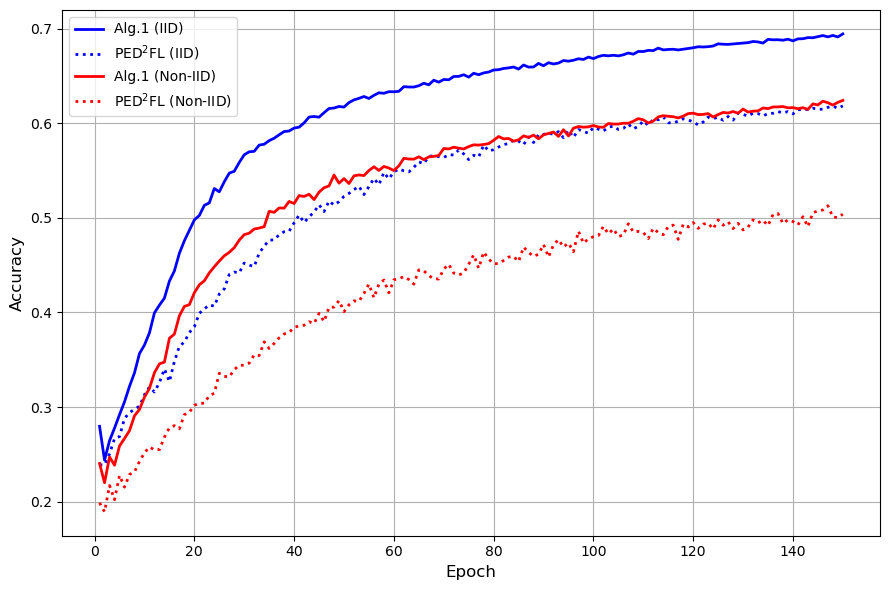} % <-- replace filename
    \vspace{-2ex}
    \caption{Average accuracy of the proposed algorithm and PED$^2$FL in a fully connected 4-client network under IID and non-IID data ($\epsilon_{\max}=1$).}
    \label{fig:Comparison with PED2FL with epsilon = 0.5}
    \vspace{-2ex}
\end{figure}
% ============================
% Figure  — Comparison with PED2FL with epsilon = 0.5
% ============================
\begin{figure}[t]
    \centering
    \includegraphics[width=0.9\linewidth]{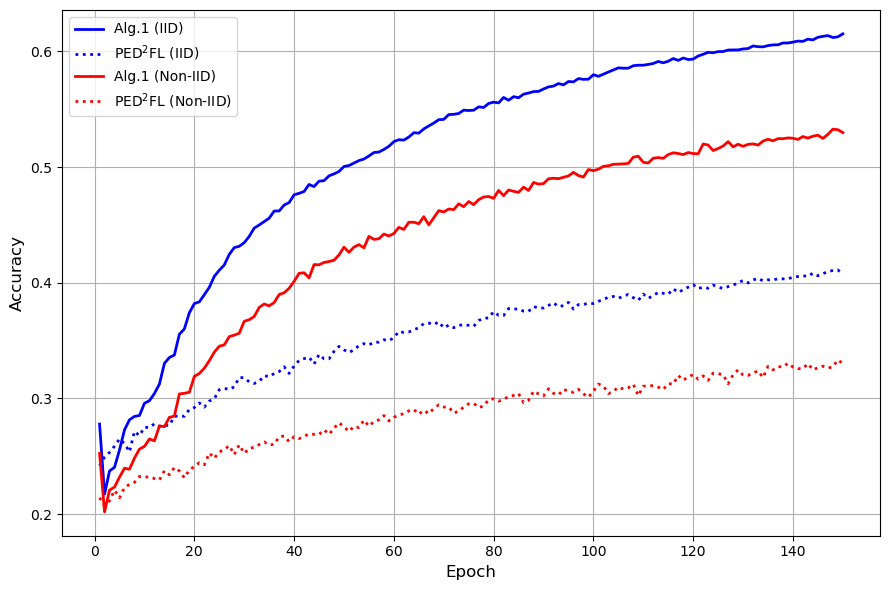} % <-- replace filename
    \vspace{-2ex}
    \caption{Average accuracy of the proposed algorithm and PED$^2$FL in a fully connected 4-client network under IID and non-IID data ($\epsilon_{\max}=0.5$).} 
    \label{fig:Comparison with PED2FL_epsilon half}
    \vspace{-2ex}
\end{figure}

To provide a comprehensive evaluation of the proposed framework, we next compare all considered combinations of network sizes, communication topologies, data distributions, and privacy levels in terms of both learning performance and cumulative privacy leakage.

% For fairness and consistency across different topologies, a single set of fixed channel gains $|h_{ij}|$ is first randomly sampled from the uniform distribution $\mathcal{U}[0.3,1]$ for the links of the fully connected graph, with $|h_{ij}|\neq |h_{ji}|$ in general. The lower bound of this range avoids excessively weak channels that could substantially suppress inter-node communication. The channel gains associated with the links retained after pruning the fully connected graph to construct the ring or random Erd\H{o}s--R\'enyi topology are inherited from the same channel realization. Therefore, the evaluated topologies share the same channel gains on their common links, allowing their performance to be compared under consistent wireless channel conditions.
For consistent comparison across topologies, fixed channel gains $|h_{ij}|$ are sampled from $\mathcal{U}[0.3,1]$ for the fully connected graph, with $|h_{ij}|\neq|h_{ji}|$ in general. The lower bound avoids excessively weak links. The ring and random topologies inherit the channel gains of their retained links from the same realization, ensuring consistent channel conditions across topologies.

Before executing the proposed DP-DL algorithm, the maximization problem in \eqref{eq:alpha-LP} is solved to obtain the optimal power-splitting coefficients $\{\alpha_i,\beta_i\}_{i=1}^{K}$, which are then kept fixed throughout training. In all simulations, the clients have identical maximum transmit power budgets, with $P_i=1$ for all $i\in\mathcal{V}$. The learning rate and Gaussian noise standard deviation are set to $\gamma_t=0.1/\sqrt{t}$ and $\sigma_{i,t}=1/\sqrt{t}$, respectively, while unit-norm gradient clipping is used, yielding $G=1$. We set $\delta=\bar{\delta}=10^{-4}$ for the per-iteration DP guarantee
and the cumulative RDP-based privacy budget, respectively. These settings are kept consistent across the experiments. To determine the training horizon for cumulative privacy accounting, we use a stopping criterion with a $15$-iteration grace period. Let $\mathrm{acc}_t$ denote the average test accuracy at iteration $t$. Training is considered stable when
$|\mathrm{acc}_{t+k}-\mathrm{acc}_t|\leq0.001$ for all $k=1,\ldots,15$. The resulting training duration determines the communication rounds used for cumulative RDP accounting.

Figs.~\ref{fig:Accuracy vs privacy for 4 clients}--\ref{fig:Accuracy vs privacy for 20 clients} show average test accuracy versus cumulative RDP leakage for $K=4,10,$ and $20$ across the considered settings. Here, $\epsilon_{\max}\in\{0.5,1,\infty\}$ is the per-iteration privacy threshold. The $48{,}000$ training samples yield about $12{,}000$, $4{,}800$, and $2{,}400$ samples per client, corresponding to subsampling rates $q\approx0.0213$, $0.0533$, and $0.1067$ with batch size $256$. These rates are used for cumulative RDP accounting.

Across all configurations, IID data generally achieve higher accuracy than non-IID data, while stronger privacy ($\epsilon_{\max}=0.5$) reduces accuracy relative to $\epsilon_{\max}=1$ and $\infty$. This reflects the privacy-utility trade-off, since stronger privacy requires more transmit power to be allocated to DP noise. Fully connected networks generally outperform ring and random topologies, as their higher connectivity improves information aggregation during training. Cumulative RDP leakage generally increases with $K$ because the fixed mini-batch size results in larger $q$ and weaker privacy amplification, while for a fixed number of clients, differences mainly reflect the number of training iterations.
% ===========================
% Figure  — Accuracy vs. Privacy for 4 nodes(Total RDP)
% ============================
\begin{figure}[t]
    \centering
    \includegraphics[width=0.9\linewidth]{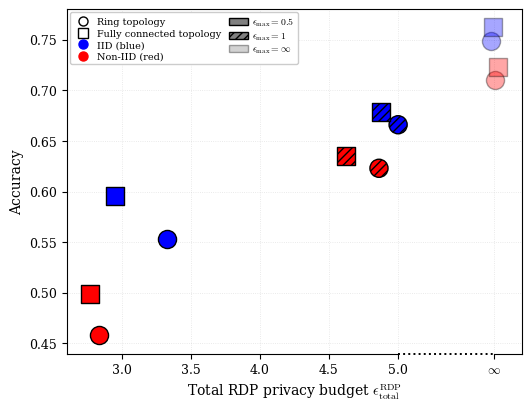} % <-- replace filename
    \vspace{-2ex}
    \caption{Average accuracy versus cumulative RDP leakage for 4 clients across topologies and data distributions, with per-iteration privacy budget $\epsilon_{\max}$.}
    \label{fig:Accuracy vs privacy for 4 clients}
    \vspace{-2ex}
\end{figure}

% Figure  — Accuracy vs. Privacy for 10 nodes (Total RDP)
% ============================
\begin{figure}[t]
    \centering
    \includegraphics[width=0.9\linewidth]{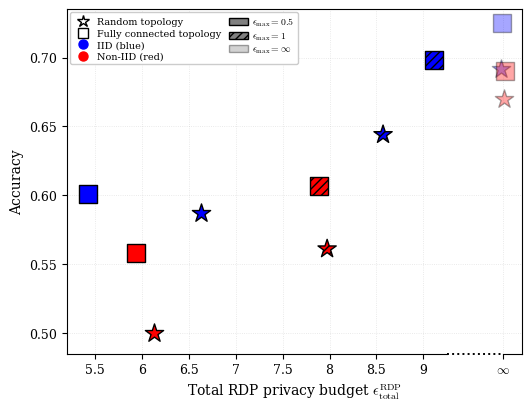} % <-- replace filename
    \vspace{-2ex}
    \caption{Average accuracy versus cumulative RDP leakage for 10 clients across topologies and data distributions, with per-iteration privacy budget $\epsilon_{\max}$.}
    \label{fig:Accuracy vs privacy for 10 clients}
    \vspace{-2ex}
\end{figure}

% Figure  — Accuracy vs. Privacy for 20 nodes (Total RDP)
% ============================
\begin{figure}[t]
    \centering
    \includegraphics[width=0.9\linewidth]{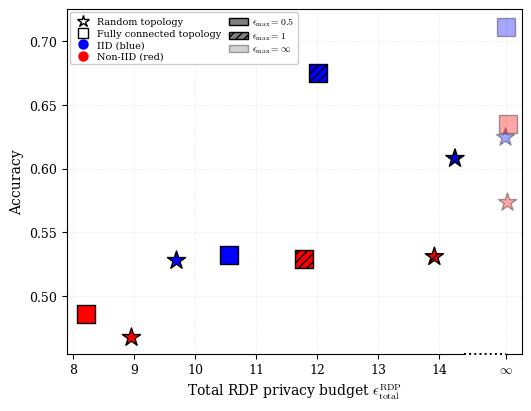} % <-- replace filename
    \vspace{-2ex}
    \caption{Average accuracy versus cumulative RDP leakage for 20 clients across topologies and data distributions, with per-iteration privacy budget $\epsilon_{\max}$.}
    \label{fig:Accuracy vs privacy for 20 clients}
    \vspace{-2ex}
\end{figure}

%%%%%%%%%%%%%%%%%%%%%%%%%%%%%%%%%%%%%%%%%%%%%%%%%%%%
\section{Conclusion}\label{sec:conclusion}
This paper proposed a power-controlled DP-DL framework for wireless multicast networks with asymmetric channel gains. Using a row-stochastic adjacency matrix and jointly allocating transmit power between model updates and Gaussian DP noise, the proposed algorithm provides both convergence and DP guarantees. We established an $O(\log T)$ regret bound, demonstrating that learning remains efficient despite channel heterogeneity and privacy-induced perturbations. Experiments on CIFAR-10 demonstrate robust convergence across different numbers of clients, data distributions, privacy budgets, and network topologies. The results highlight the privacy-utility trade-off and the role of network connectivity in mitigating the effects of DP noise, demonstrating the effectiveness of the proposed DP-DL framework in realistic wireless settings.
% -----------------------------------------------------------
\appendices
\section{Proof of Convergence}\label[appendix]{app:proof}
Denote $\mathbf{s}_{i,t}=\sum_{j=1}^{K} a_{ij}(\mathbf{x}_{j,t}+\sqrt{\frac{\beta_{j}}{\alpha_{j}}} \boldsymbol{\eta}_{j,t})$, $\mathbf{w}_{i,t}=\sum_{j=1}^{K} a_{ij}\mathbf{x}_{j,t}$ and $\mathbf{q}_{i,t+1}=\mathbf{x}_{i,t+1}-\mathbf{s}_{i,t}$. Since $\mathbf{x}_{j,t}\in\Omega$, their convex combination $\mathbf{w}_{i,t}\in\Omega$, implying $\|\Pi_{\Omega}(\mathbf{y})-\mathbf{w}_{i,t}\|\leq\|\mathbf{y}-\mathbf{w}_{i,t}\|$ for any $\mathbf{y}$, $i$, and $t$.
We have:
%\footnote{\jj{Y: Now I see why my suggested bound doesn't work, but adding and subtracting $\mathbf w_{i,t}$ does: because $\mathbf w_{i,t} \in \Omega$ and $\mathbf s_{i,t} \notin \Omega$ in general.}}
\begin{equation}\label{q_i,t+1upperbound}
\begin{aligned}
    &\|\mathbf{q}_{i,t+1}\| = \|\mathbf{x}_{i,t+1} - \mathbf{s}_{i,t}\| \\
    &\leq \|\mathbf{x}_{i,t+1} - \mathbf{w}_{i,t}\| + \|\mathbf{w}_{i,t} - \mathbf{s}_{i,t}\|\\
   %  &\leq \|\mathbf{s}_{i,t}-\frac{\gamma_{t} \mathbf{g}_{i,t}}{z_{ii,t}} -\mathbf{w}_{i,t}\| + \|\mathbf{w}_{i,t} - \mathbf{s}_{i,t}\|\\
   % %&= \|\mathbf{s}_{i,t}-(\frac{\gamma_{t} \mathbf{g}_{i,t}}{z_{ii,t}} +\mathbf{w}_{i,t})\| + \|\mathbf{w}_{i,t} - \mathbf{s}_{i,t}\|\\
   %  &\leq \|\mathbf{s}_{i,t} - \mathbf{w}_{i,t}\| + \|\mathbf{w}_{i,t}-(\frac{\gamma_{t} \mathbf{g}_{i,t}}{z_{ii,t}} + \mathbf{w}_{i,t})\| + \|\mathbf{w}_{i,t} - \mathbf{s}_{i,t}\|\\
    &\leq 2\|\mathbf{s}_{i,t}-\mathbf{w}_{i,t}\| + \gamma_{t}\|\frac{\mathbf{g}_{i,t}}{z_{ii,t}}\| \leq 2\sum_{j=1}^{K}  \| \sqrt{\frac{\beta_{j}}{\alpha_{j}}}\boldsymbol{\eta}_{j,t} \|  + \gamma_{t} \theta G ,
\end{aligned}
\end{equation}
where the last inequality follows by assuming that all $a_{ij}$ are set to their maximum value of $1$.

Recursively expanding $\mathbf{s}_{i,t}$ yields
\begin{equation}\label{x_i_t+1_in_terms_of_q}
\begin{aligned}
    \mathbf{x}_{i,t+1}& = \mathbf{q}_{i,t+1}+\sum_{l=1}^{t}\sum_{j=1}^{K} [A^{t-l+1}]_{ij}\mathbf{q}_{j,l} \\
    &+ \sum_{j=1}^{K} [A^{t+1}]_{ij} \mathbf x_{j,0} + \sum_{l=0}^{t}\sum_{j=1}^{K} [A^{t-l+1}]_{ij}\sqrt{\frac{\beta_{j}}{\alpha_{j}}}\boldsymbol{\eta}_{j,l}.
\end{aligned}    
\end{equation}

%For calculating $\bar{\mathbf{x}}_{t+1}$, the expression $\mathbf{q}_{i,t+1} = \mathbf{x}_{i,t+1} - \mathbf{s}_{i,t}$ is utilized as follows:
% Next, we define the weighted sum of individual model vectors:
% \begin{equation}\label{x_bar}
% \begin{aligned}
%     \bar{\mathbf{x}}_{t+1}
%     &\triangleq \sum_{i=1}^{K} \pi_i \mathbf{x}_{i,t+1} = \sum_{i=1}^{K}\pi_i (\mathbf{q}_{i,t+1} + \mathbf{s}_{i,t})\\
%     &= \sum_{i=1}^{K}\pi_i 
%     \Big(\mathbf{q}_{i,t+1}+\sum_{j=1}^{K} a_{ij} (\mathbf{x}_{j,t}+\sqrt{\frac{\beta_{j}}{\alpha_{j}}}\boldsymbol{\eta}_{j,t})\Big)\\
%     %&\jj{= \sum_{i=1}^{K}\pi_i \mathbf{q}_{i,t+1} + \sum_{j=1}^K \left(\sum_{i=1}^{K}\pi_i a_{ij} \right) \mathbf{x}_{j,t}}\\
%     %&\jj{\qquad+\sum_{j=1}^K \left(\sum_{i=1}^{K}\pi_i a_{ij} \right) \sqrt{\frac{\beta_{j}}{\alpha_{j}}}\boldsymbol{\eta}_{j,t}}\\
%     &\overset{\text{(a)}}{=} \sum_{i=1}^{K}\pi_i \mathbf{q}_{i,t+1} + \sum_{i=1}^{K}\pi_i \mathbf{x}_{i,t}
%     +\sum_{i=1}^{K}\pi_i \sqrt{\frac{\beta_{i}}{\alpha_{i}}}\boldsymbol{\eta}_{i,t}\\
%     &= \sum_{i=1}^{K}\pi_i \mathbf{q}_{i,t+1} + \bar{\mathbf{x}}_{t} +
%     \sum_{i=1}^{K}\pi_i \sqrt{\frac{\beta_{i}}{\alpha_{i}}}\boldsymbol{\eta}_{i,t},\\
% \end{aligned}    
% \end{equation}
% where (a) comes from the fact that $\boldsymbol{\pi}^\top A = \boldsymbol{\pi}^\top$, and thus, $\sum_{i=1}^K \pi_i a_{ij} = \pi_j$.
Considering the weighted average as
$\bar{\mathbf{x}}_t=\sum_{i=1}^{K}\pi_i\mathbf{x}_{i,t}$ and using $\boldsymbol{\pi}^{\top}A=\boldsymbol{\pi}^{\top}$, we obtain
\begin{equation}\label{x_bar}
\bar{\mathbf{x}}_{t+1}
=
\bar{\mathbf{x}}_t
+\sum_{i=1}^{K}\pi_i\mathbf{q}_{i,t+1}
+\sum_{i=1}^{K}\pi_i
\sqrt{\frac{\beta_i}{\alpha_i}}\boldsymbol{\eta}_{i,t}.
\end{equation}

%By taking the summation of both sides of the equation \cref{x_bar} from 0 to $t$, the following can be obtained:
%\begin{equation}\label{summation_x_bar}
%\begin{aligned}
%    \sum_{l=0}^{t} \bar{\mathbf{x}}_{l+1} &= \bar{\mathbf{x}}_{1} + \bar{\mathbf{x}}_{2} + ...+\bar{\mathbf{x}}_{t+1}\\
%    &= \sum_{l=1}^{t+1} \sum_{i=1}^{K}\pi_i \mathbf{q}_{i,l}  + \bar{\mathbf{x}}_{0} + \bar{\mathbf{x}}_{1} + ...+\bar{\mathbf{x}}_{t} \\
%    &+ \sum_{l=0}^{t}\sum_{i=1}^{K}\pi_i \sqrt{\frac{\beta_{i}}{\alpha_{i}}}\boldsymbol{\eta}_{i,l}.\\
%\end{aligned}
%\end{equation}

%By canceling the common terms $\bar{\mathbf{x}}_1,\ldots,\bar{\mathbf{x}}_t$ on both sides of \eqref{summation_x_bar}, we obtain the following expression for $\bar{\mathbf{x}}_{t+1}$.

Summing \cref{x_bar} from $l=0$ to $t$ and telescoping yields
\begin{equation}\label{final_x_bar}
\begin{aligned}
    \bar{\mathbf{x}}_{t+1}=\bar{\mathbf{x}}_{0}+\sum_{l=1}^{t+1} \sum_{i=1}^{K}\pi_i \mathbf{q}_{i,l}+\sum_{l=0}^{t}\sum_{i=1}^{K}\pi_i \sqrt{\frac{\beta_{i}}{\alpha_{i}}}\boldsymbol{\eta}_{i,l}.
\end{aligned}
\end{equation}

Combining (18), (19), and (21) with Lemma~1 and applying the
triangle inequality, we obtain the following expected consensus error bound:
\begin{equation}\label{formula 24}
\begin{aligned}
    &\mathbb{E} \|\bar{\mathbf{x}}_{t+1} - \mathbf{x}_{i,t+1}\|\\
    &\leq 4\sum_{j=1}^{K} \mathbb{E}\|\sqrt{\frac{\beta_{j}}{\alpha_{j}}}\boldsymbol{\eta}_{j,t}\| + 2\gamma_{t} \theta G + C\xi^{t+1}\sum_{j=1}^{K} \|\mathbf{x}_{j,0}\|\\
    &+ C\sum_{l=0}^{t}\xi^{t-l+1}\sum_{j=1}^{K}\mathbb{E}\|\sqrt{\frac{\beta_{j}}{\alpha_{j}}}\boldsymbol{\eta}_{j,l} \| \\
    &+2KC\sum_{l=1}^{t}\xi^{t-l+1}\sum_{j=1}^{K}\mathbb{E}\| \sqrt{\frac{\beta_{j}}{\alpha_{j}}}\boldsymbol{\eta}_{j,l-1}\|\\
    &+ K\theta C G \sum_{l=1}^{t} \xi^{t-l+1}\gamma_{l-1}.\\ 
\end{aligned}    
\end{equation}

Now, we want to find an upperbound for $\sum_{i=1}^{K} \pi_i \mathbb{E} \|\mathbf{x}_{i,t+1} - \mathbf{x}\|^{2}$ in order to utilize in the next steps. First, we start with $\|\mathbf{x}_{i,t+1} - \mathbf{x} \| ^{2}$ as follows:
\begin{equation}\label{||x_i,t+1-x||^2}
\begin{aligned}
    &\|\mathbf{x}_{i,t+1} - \mathbf{x} \| ^ {2} \leq  \|\mathbf{s}_{i,t} - \gamma_{t} \frac{\mathbf{g}_{i,t}}{z_{ii,t}} - \mathbf{x}\|^{2} \\
    & = \|\mathbf{s}_{i,t} - \mathbf{x}\|^{2} + \frac{\gamma_{t}^{2}}{z_{ii,t}^{2}} \|\mathbf{g}_{i,t}\|^{2} - 2 \frac{\gamma_{t}}{z_{ii,t}} \mathbf{g}_{i,t}^{T}(\mathbf{s}_{i,t} - \mathbf{x}). \\
\end{aligned}    
\end{equation}

The three terms in \eqref{||x_i,t+1-x||^2} can be bounded using the
zero-mean Gaussian noise, \cref{ass:bounded-gradients}, and
\cref{ass:strong-convexity}, respectively, as follows:
\begin{equation}\label{E||s_i,t-x||^2}
\begin{aligned}
\mathbb{E}\|\mathbf{s}_{i,t}-\mathbf{x}\|^2
\leq\;&
\sum_{j=1}^{K}a_{ij}\mathbb{E}\|\mathbf{x}_{j,t}-\mathbf{x}\|^2
+\sum_{j=1}^{K}a_{ij}
\mathbb{E}\left\|
\sqrt{\frac{\beta_j}{\alpha_j}}\boldsymbol{\eta}_{j,t}
\right\|^2 .
\end{aligned}
\end{equation}
Moreover,
\begin{equation}\label{gamma/z_ii,t||g_i,t||^2}
\frac{\gamma_t^2}{z_{ii,t}^2}\|\mathbf{g}_{i,t}\|^2
\leq\gamma_t^2\theta^2G^2,
\end{equation}
and
\begin{equation}\label{-g_i,t(s_i,t-x)upperbound}
\begin{aligned}
&-\mathbf{g}_{i,t}^{T}(\mathbf{s}_{i,t}-\mathbf{x})\\
&\leq
\|\mathbf{g}_{i,t}\|\sum_{j=1}^{K}a_{ij}
\left\|\sqrt{\frac{\beta_j}{\alpha_j}}\boldsymbol{\eta}_{j,t}\right\|
+\|\mathbf{g}_{i,t}\|\sum_{j=1}^{K}a_{ij}
\|\mathbf{x}_{j,t}-\bar{\mathbf{x}}_t\|\\
&+\|\mathbf{g}_{i,t}\|\|\bar{\mathbf{x}}_t-\mathbf{x}_{i,t}\|
+f_i(\mathbf{x})-f_i(\mathbf{x}_{i,t})
-\frac{\mu}{2}\|\mathbf{x}-\mathbf{x}_{i,t}\|^2.
\end{aligned}
\end{equation}

Finally, the following inequality is obtained by combining \eqref{E||s_i,t-x||^2}, \eqref{gamma/z_ii,t||g_i,t||^2}, and \eqref{-g_i,t(s_i,t-x)upperbound} and applying \cref{ass:bounded-gradients}:
\begin{equation}\label{E||x_i,t+1-x||^2}
\begin{aligned}
    &\mathbb{E} \|\mathbf{x}_{i,t+1} - \mathbf{x}\|^{2}\\ 
    &\leq \sum_{j=1}^{K} a_{ij} \mathbb{E} \|\mathbf{x}_{j,t} - \mathbf{x} \|^{2} + \sum_{j=1}^{K} a_{ij} \mathbb{E} \|\sqrt{\frac{\beta_{j}}{\alpha_{j}}}\boldsymbol{\eta}_{j,t}\|^{2} \\
    &+ \gamma_{t}^{2} \theta^{2} G^{2} + 2\theta \gamma_{t}(G \sum_{j=1}^{K} a_{ij} \mathbb{E} \|\sqrt{\frac{\beta_{j}}{\alpha_{j}}}\boldsymbol{\eta}_{j,t}\|\\
    &+ G \sum_{j=1}^{K} a_{ij} \mathbb{E} \| \mathbf{x}_{j,t} - \bar{\mathbf{x}}_t\| + G \|\bar{\mathbf{x}}_t - \mathbf{x}_{i,t}\| \\
    &- \frac{\mu}{2} \mathbb{E} \|\mathbf{x} - \mathbf{x}_{i,t}\|^{2}) + \frac{2\gamma_{t}}{z_{ii,t}}\mathbb{E} [f_{i}(\mathbf{x}) - f_{i}(\mathbf{x}_{i,t})].
\end{aligned}    
\end{equation}

By multiplying both sides of \eqref{E||x_i,t+1-x||^2} by $\pi_i$ and then summing over $i$, we obtain the following:
\begin{equation}\label{summation_pi_i_E||x_i,t+1-x||^2}
\begin{aligned}
    &\sum_{i=1}^{K} \pi_i \mathbb{E} \|\mathbf{x}_{i,t+1} - \mathbf{x}\|^{2}\\ 
    % &\leq \sum_{i=1}^{K} \pi_{i} \mathbb{E} \|\mathbf{x}_{i,t} - \mathbf{x} \|^{2} + \sum_{i=1}^{K} \pi_{i} \mathbb{E} \|\sqrt{\frac{\beta_{i}}{\alpha_{i}}}\boldsymbol{\eta}_{i,t}\|^{2} \\
    % &+ 2\theta G \gamma_t \sum_{i=1}^{K} \pi_i \mathbb{E} \|\sqrt{\frac{\beta_{i}}{\alpha_{i}}}\boldsymbol{\eta}_{i,t}\|\\ 
    % &+ 2\theta G \gamma_t \sum_{i=1}^{K} \pi_i \mathbb{E} \|\mathbf{x}_{i,t} - \bar{\mathbf{x}}_t\| + 2\theta G \gamma_t \sum_{i=1}^{K} \pi_i \mathbb{E} \| \bar{\mathbf{x}}_t - \mathbf{x}_{i,t}\|\\
    % & - 2 \theta \gamma_t \sum_{i=1}^{K} \frac{\mu}{2} \pi_i \mathbb{E} \|\mathbf{x} - \mathbf{x}_{i,t}\|^{2} + \theta^{2} G^{2} \gamma_{t}^{2} \sum_{i=1}^{K} \pi_i \\
    % &+ 2 \gamma_t \sum_{i=1}^{K} \frac{\pi_i}{z_{ii,t}} \mathbb{E} [f_{i}(\mathbf{x}) - f_{i}(\mathbf{x}_{i,t})]\\
    % &=
    &\leq (1-\mu \theta \gamma_t ) \sum_{i=1}^{K} \pi_{i} \mathbb{E} \|\mathbf{x}_{i,t} - \mathbf{x} \|^{2} + \sum_{i=1}^{K} \pi_{i} \mathbb{E} \|\sqrt{\frac{\beta_{i}}{\alpha_{i}}}\boldsymbol{\eta}_{i,t}\|^{2} \\
    &+ 2\theta G \gamma_t \sum_{i=1}^{K} \pi_i \mathbb{E} \|\sqrt{\frac{\beta_{i}}{\alpha_{i}}}\boldsymbol{\eta}_{i,t}\| + 4\theta G \gamma_t \sum_{i=1}^{K} \pi_i \mathbb{E} \|\mathbf{x}_{i,t} - \bar{\mathbf{x}}_t\|\\
    & + \theta^{2} G^{2} \gamma_{t}^{2} + 2 \gamma_t \sum_{i=1}^{K} \frac{\pi_i}{z_{ii,t}} \mathbb{E} [f_{i}(\mathbf{x}) - f_{i}(\mathbf{x}_{i,t})].\\
\end{aligned}    
\end{equation}

Applying Jensen's inequality and utilizing \cref{lem:convergence-to-pi}, we can rewrite the last term in \eqref{summation_pi_i_E||x_i,t+1-x||^2} as follows:
\begin{equation}\label{pi_i/z_ii,t_E[f_i,t(x)-f_i,t(x_i,t)]}
\begin{aligned}
    &\frac{\pi_i}{z_{ii,t}} \mathbb{E} [f_{i}(\mathbf{x}) - f_{i}(\mathbf{x}_{i,t})]\\ 
    &= (\frac{\pi_i}{z_{ii,t}} -1) \mathbb{E} [f_{i}(\mathbf{x}) - f_{i}(\mathbf{x}_{i,t})] + \mathbb{E} [f_{i}(\mathbf{x}) - f_{i}(\mathbf{x}_{i,t})]\\
    &\leq |\frac{\pi_i - z_{ii,t}}{z_{ii,t}}| \mathbb{E} [|f_{i}(\mathbf{x}) - f_{i}(\mathbf{x}_{i,t})|] + \mathbb{E} [f_{i}(\mathbf{x}) - f_{i}(\mathbf{x}_{i,t})]\\
    &\leq \frac{C\xi^{t}}{|z_{ii,t}|}\mathbb{E} [|f_{i}(\mathbf{x}) - f_{i}(\mathbf{x}_{i,t})|] + \mathbb{E} [f_{i}(\mathbf{x}) - f_{i}(\mathbf{x}_{i,t})].
\end{aligned}    
\end{equation}

Using \eqref{pi_i/z_ii,t_E[f_i,t(x)-f_i,t(x_i,t)]}, and then rearranging the resulting bound and evaluating it at $\mathbf{x}=\mathbf{x}^\ast$, we obtain
\begin{equation}\label{summation_E[f_i,t(x)-f_i,t(x*)]}
\begin{aligned}
&\sum_{i=1}^{K} \mathbb{E}\!\left[
f_{i}(\mathbf{x}_{i,t}) - f_{i}(\mathbf{x}^*)
\right]\\
&\le
\frac{1}{2\gamma_t}
\left(1-\mu\theta\gamma_t\right)
\sum_{i=1}^{K} \pi_i
\mathbb{E}\!\left[
\|\mathbf{x}_{i,t}-\mathbf{x}^*\|^2
\right]
\\
&-\frac{1}{2\gamma_t}
\sum_{i=1}^{K} \pi_i
\mathbb{E}\!\left[
\|\mathbf{x}_{i,t+1}-\mathbf{x}^*\|^2
\right]
+\frac{1}{2\gamma_t}
\sum_{i=1}^{K} \pi_i
\mathbb{E}
\|
\sqrt{\tfrac{\beta_i}{\alpha_i}}\,
\boldsymbol{\eta}_{i,t}
\|^2
\\
&+\theta G
\sum_{i=1}^{K} \pi_i
\mathbb{E}\|
\sqrt{\tfrac{\beta_i}{\alpha_i}}\,
\boldsymbol{\eta}_{i,t}
\|
+2\theta G
\sum_{i=1}^{K} \pi_i
\mathbb{E}
\|\mathbf{x}_{i,t}-\bar{\mathbf{x}}_t\|
\\
& +\tfrac{1}{2}\theta^2 G^2 \gamma_t
+\theta C G \xi^{t}
\sum_{i=1}^{K}
\mathbb{E}
\|\mathbf{x}_{i,t}-\mathbf{x}^*\|.
\end{aligned}
\end{equation}

Since the regret at node $j$ involves $f_i(x_{j,t})$, using the Lipschitz continuity of $f_i(\cdot)$ and the bounded subgradient, we obtain
\begin{equation}\label{f(x_j,t)-f(x_i,t}
\begin{aligned}
f_i(x_{j,t})-f_i(x^\star)
&\leq f_i(x_{i,t})-f_i(x^\star)\\
&\quad+G\left(\|x_{j,t}-\bar{x}_t\|
+\|x_{i,t}-\bar{x}_t\|\right).
\end{aligned}
\end{equation}

By taking into consideration \eqref{f(x_j,t)-f(x_i,t} and then summing \eqref{summation_E[f_i,t(x)-f_i,t(x*)]} over $T$ iterations, we can find an upper bound on the expected regret function $\mathbb{E}[\mathbb{R}_{j}(T)]$, as follows:
\begin{equation}\label{E[R_i(T)]}
\begin{aligned}
&\sum_{t=1}^{T}\sum_{i=1}^{K} \mathbb{E}\!\left[
f_{i}(\mathbf{x}_{j,t}) - f_{i}(\mathbf{x}^*)
\right]\\
&\quad
\le
\sum_{t=1}^{T}
\frac{1}{2\gamma_t}
\left(1-\mu\theta\gamma_t\right)
\sum_{i=1}^{K} \pi_i
\mathbb{E}\!\left[
\|\mathbf{x}_{i,t}-\mathbf{x}^*\|^2
\right]
\\
&\quad
-\sum_{t=1}^{T}
\frac{1}{2\gamma_t}
\sum_{i=1}^{K} \pi_i
\mathbb{E}\!\left[
\|\mathbf{x}_{i,t+1}-\mathbf{x}^*\|^2
\right]
\\
&\quad
+\sum_{t=1}^{T}
\frac{1}{2\gamma_t}
\sum_{i=1}^{K} \pi_i
\mathbb{E}
\Big\|
\sqrt{\tfrac{\beta_i}{\alpha_i}}\,
\boldsymbol{\eta}_{i,t}
\Big\|^2
\\
&\quad
+\sum_{t=1}^{T}
\theta G
\sum_{i=1}^{K} \pi_i
\mathbb{E}
\Big\|
\sqrt{\tfrac{\beta_i}{\alpha_i}}\,
\boldsymbol{\eta}_{i,t}
\Big\|
\!+\!\sum_{t=1}^{T}
2\theta G
\sum_{i=1}^{K} \pi_i
\mathbb{E}
\|\mathbf{x}_{i,t}-\bar{\mathbf{x}}_t\|
\\
&\quad
+\sum_{t=1}^{T}
\tfrac{1}{2}\theta^2 G^2 \gamma_t
+\sum_{t=1}^{T}
\theta C G \xi^{t}
\sum_{i=1}^{K}
\mathbb{E}
\|\mathbf{x}_{i,t}-\mathbf{x}^*\| \\
&\quad + KG
\sum_{t=1}^{T}
\mathbb{E}\!\left[
\left\|
\mathbf{x}_{j,t}
-
\bar{\mathbf{x}}_{t}
\right\|
\right]
+
G
\sum_{t=1}^{T}
\sum_{i=1}^{K}
\mathbb{E}\!\left[
\left\|
\mathbf{x}_{i,t}
-
\bar{\mathbf{x}}_{t}
\right\|
\right].
\end{aligned}
\end{equation}

Now, we derive an upperbound for the noise-related terms in \eqref{E[R_i(T)]} as follows: 
\begin{equation}\label{E||sqrt(beta_i/alpha_i)eta_i,t|| _upperbound}
\begin{aligned}
    & \mathbb{E}[\|\sqrt{\frac{\beta_{i}}{\alpha_{i}}} \boldsymbol{\eta}_{i,t} \|] \leq \sqrt{m} \max_{i} \sqrt{\frac{\beta_{i}}{\alpha_{i}}} \sigma_{i,t}\\
    &\stackrel{(a)}{\leq} \frac{\sqrt{m} \max_{i} \Big ( \sqrt{\frac{\beta_{i}}{\alpha_{i}}}\Big ) G \gamma_t \theta \max_{i,j} \Big ( 2|h_{ji}| \sqrt{\alpha_j P_j} \Big )  \sqrt{2 \ln \frac{1.25}{\delta}} }{ \min_{i,j,k} \Big( \epsilon_{ij} \sqrt{\sum_{k \in N_i} |h_{ki}|^2 \beta_k P_k }\Big)},\\
\end{aligned}
\end{equation}
where (a) is obtained by using the expression of $\epsilon_{ij}$ in \eqref{eq:e_ij-thm}. 
% Furthermore, $m$ is the dimension of both model coefficients and noise converted to a vector. 

Also, we can bound $\mathbb{E}[\|\sqrt{\frac{\beta_{i}}{\alpha_{i}}} \boldsymbol{\eta}_{i,t} \|^{2}]$ as follows:
\begin{equation}\label{summation_E[||sqrt_beta_i/alpha_i_eta_i,t||^2]}
\begin{aligned}
    &\mathbb{E}[\|\sqrt{\frac{\beta_{i}}{\alpha_{i}}} \boldsymbol{\eta}_{i,t} \|^{2}] \leq m \max_{i} {\frac{\beta_{i}}{\alpha_{i}}} \sigma_{i,t}^{2}\\
    &\leq \frac{m \max_{i} \Big ({\frac{\beta_{i}}{\alpha_{i}}}\Big ) G^{2} \gamma_{t}^{2} \theta^{2} \max_{i,j} \Big ( 2 |h_{ji}| \sqrt{\alpha_j P_j} \Big )^{2}  \Big ({2 \ln \frac{1.25}{\delta}}\Big ) }{ \min_{i,j,k} \Big( \epsilon_{ij}^{2} {\sum_{k \in N_i} |h_{ki}|^2 \beta_k P_k }\Big)}.\\
\end{aligned}
\end{equation}

By considering $\gamma_t = \frac{1}{\mu \theta t}$, we have:
\begin{equation}\label{summation_gamma}
    \begin{aligned}
\sum_{t=1}^{T} \gamma_t = \sum_{t=1}^{T} \frac{1}{\theta \mu\, t}
\;\le\;
\frac{1}{\theta \mu}\,(1+\log T).
    \end{aligned}
\end{equation}

Using \eqref{formula 24} and substituting the results from
\eqref{summation_gamma},
\eqref{E||sqrt(beta_i/alpha_i)eta_i,t|| _upperbound}, and
\eqref{summation_E[||sqrt_beta_i/alpha_i_eta_i,t||^2]}
into \eqref{E[R_i(T)]}, we obtain the following upper bound on the expected regret $\mathbb{E}[\mathbb{R}_{j}(T)]$:
\begin{equation}\label{Final upperbound}
\begin{aligned}
    &\mathbb{E}[\mathbb{R}_{j}(T)] \leq \frac{\xi C G}{1-\xi} \Big (2(K+\theta) \sum_{i=1}^{K}\|x_{i,0}\| + K\theta L \Big ) \\
    & + (1 + \log T) \Big (\frac{\max_{i} \sqrt{\frac{\beta_{i}}{\alpha_{i}}}\max_{i,j} \big ( |h_{ji}| \sqrt{\alpha_j P_j} \big )}{\mu} \Big )\\ 
    &\times \Big (\frac{2 (8K^{2}+8\theta K + \theta )\sqrt{m} G^{2}  \sqrt{2 \ln \frac{1.25}{\delta}} }{\min_{i,j,k} \Big( \epsilon_{ij} \sqrt{\sum_{k \in N_i} |h_{ki}|^2 \beta_k P_k }\Big)} \Big)\\
    &+ (1 + \log T) \Big ( \frac{\max_{i} \Big ({\frac{\beta_{i}}{\alpha_{i}}}\Big )  \max_{i,j} \Big ( |h_{ji}| \sqrt{\alpha_j P_j} \Big )^{2}}{\mu} \Big ) \\
    &\times \Big ( \frac{4 m G^{2} \theta \ln (\frac{1.25}{\delta})}{\min_{i,j,k} \Big( \epsilon_{ij} \sqrt{\sum_{k \in N_i} |h_{ki}|^2 \beta_k P_k }\Big)} \Big) \\ 
    &+ (1 + \log T) \Big (\frac{\max_{i} \sqrt{\frac{\beta_{i}}{\alpha_{i}}}\max_{i,j} \big ( |h_{ji}| \sqrt{\alpha_j P_j} \big )}{\mu} \Big )\\ 
    &\times \frac{4(2K+1)(\theta + K) K \sqrt{m} C G^{2} \sqrt{2 \ln \frac{1.25}{\delta}}}{(1-\xi)\min_{i,j,k} \Big( \epsilon_{ij} \sqrt{\sum_{k \in N_i} |h_{ki}|^2 \beta_k P_k }\Big)}\\
    &+ (1 + \log T) \frac{2K (K + \theta) C G^{2}}{1-\xi} + (1 + \log T) \frac{1}{2}(8K+9\theta)G^{2}.
\end{aligned}
\end{equation}

Thus, by collecting the terms in \eqref{Final upperbound} into the constants $U_1$ and $U_2$, we obtain the bound stated in~\cref{Theorem2_convergence}, completing the proof.

% -------------------------------------------------------------
\bibliographystyle{IEEEtran}
\bibliography{References}

@ARTICLE{10542323,
  author={Yuan, Liangqi and Wang, Ziran and Sun, Lichao and Yu, Philip S. and Brinton, Christopher G.},
  journal={IEEE Internet of Things Journal}, 
  title={Decentralized Federated Learning: A Survey and Perspective}, 
  year={2024},
  volume={11},
  number={21},
  pages={34617-34638},
  doi={10.1109/JIOT.2024.3407584}}

@ARTICLE{10251949,
  author={Martínez Beltrán, Enrique Tomás and Pérez, Mario Quiles and Sánchez, Pedro Miguel Sánchez and Bernal, Sergio López and Bovet, Gérôme and Pérez, Manuel Gil and Pérez, Gregorio Martínez and Celdrán, Alberto Huertas},
  journal={IEEE Communications Surveys \& Tutorials}, 
  title={Decentralized Federated Learning: Fundamentals, State of the Art, Frameworks, Trends, and Challenges}, 
  year={2023},
  volume={25},
  number={4},
  pages={2983-3013},
  doi={10.1109/COMST.2023.3315746}}

@ARTICLE{10420449,
  author={Hallaji, Ehsan and Razavi-Far, Roozbeh and Saif, Mehrdad and Wang, Boyu and Yang, Qiang},
  journal={IEEE Transactions on Big Data}, 
  title={Decentralized Federated Learning: A Survey on Security and Privacy}, 
  year={2024},
  volume={10},
  number={2},
  pages={194-213},
  doi={10.1109/TBDATA.2024.3362191}}

@ARTICLE{9220780,
  author={Abdulrahman, Sawsan and Tout, Hanine and Ould-Slimane, Hakima and Mourad, Azzam and Talhi, Chamseddine and Guizani, Mohsen},
  journal={IEEE Internet of Things Journal}, 
  title={A Survey on Federated Learning: The Journey From Centralized to Distributed On-Site Learning and Beyond}, 
  year={2021},
  volume={8},
  number={7},
  pages={5476-5497},
  doi={10.1109/JIOT.2020.3030072}}

@ARTICLE{9460016,
  author={Khan, Latif U. and Saad, Walid and Han, Zhu and Hossain, Ekram and Hong, Choong Seon},
  journal={IEEE Communications Surveys \& Tutorials}, 
  title={Federated Learning for Internet of Things: Recent Advances, Taxonomy, and Open Challenges}, 
  year={2021},
  volume={23},
  number={3},
  pages={1759-1799},
  doi={10.1109/COMST.2021.3090430}}

@INPROCEEDINGS{9306745,
  author={Drainakis, Georgios and Katsaros, Konstantinos V. and Pantazopoulos, Panagiotis and Sourlas, Vasilis and Amditis, Angelos},
  booktitle={2020 IEEE 19th International Symposium on Network Computing and Applications (NCA)}, 
  title={Federated vs. Centralized Machine Learning under Privacy-elastic Users: A Comparative Analysis}, 
  year={2020},
  volume={},
  number={},
  pages={1-8},
  doi={10.1109/NCA51143.2020.9306745}}

@ARTICLE{9716792,
  author={Ye, Hao and Liang, Le and Li, Geoffrey Ye},
  journal={IEEE Journal of Selected Topics in Signal Processing}, 
  title={Decentralized Federated Learning With Unreliable Communications}, 
  year={2022},
  volume={16},
  number={3},
  pages={487-500},
  doi={10.1109/JSTSP.2022.3152445}}

@ARTICLE{10506083,
  author={Zhai, Zhiyuan and Yuan, Xiaojun and Wang, Xin},
  journal={IEEE Transactions on Wireless Communications}, 
  title={Decentralized Federated Learning via {MIMO} Over-the-Air Computation: Consensus Analysis and Performance Optimization}, 
  year={2024},
  volume={23},
  number={9},
  pages={11847-11862},
  doi={10.1109/TWC.2024.3385443}}

@INPROCEEDINGS{10279097,
  author={Michelusi, Nicolò},
  booktitle={ICC 2023 - IEEE International Conference on Communications}, 
  title={Decentralized Federated Learning via Non-Coherent Over-the-Air Consensus}, 
  year={2023},
  volume={},
  number={},
  pages={3102-3107},
  doi={10.1109/ICC45041.2023.10279097}}

@ARTICLE{10025677,
  author={Chen, Shuzhen and Wang, Yangyang and Yu, Dongxiao and Ren, Ju and Xu, Congan and Zheng, Yanwei},
  journal={IEEE Transactions on Computers}, 
  title={Privacy-Enhanced Decentralized Federated Learning at Dynamic Edge}, 
  year={2023},
  volume={72},
  number={8},
  pages={2165-2180},
  doi={10.1109/TC.2023.3239542}}

@INPROCEEDINGS{8433217,
  author={Gao, Huan and Zhang, Chunlei and Ahmad, Muaz and Wang, Yongqiang},
  booktitle={2018 IEEE Conference on Communications and Network Security (CNS)}, 
  title={Privacy-Preserving Average Consensus on Directed Graphs Using Push-Sum}, 
  year={2018},
  volume={},
  number={},
  pages={1-9},
  doi={10.1109/CNS.2018.8433217}}

@ARTICLE{10068288,
  author={Chen, Xiaomeng and Huang, Lingying and Ding, Kemi and Dey, Subhrakanti and Shi, Ling},
  journal={IEEE Transactions on Automatic Control}, 
  title={Privacy-Preserving Push-Sum Average Consensus via State Decomposition}, 
  year={2023},
  volume={68},
  number={12},
  pages={7974-7981},
  doi={10.1109/TAC.2023.3256479}}

@ARTICLE{10535197,
  author={Cheng, Huqiang and Liao, Xiaofeng and Li, Huaqing and Lü, Qingguo and Zhao, You},
  journal={IEEE Transactions on Signal and Information Processing over Networks}, 
  title={Privacy-Preserving Push-Pull Method for Decentralized Optimization via State Decomposition}, 
  year={2024},
  volume={10},
  number={},
  pages={513-526},
  doi={10.1109/TSIPN.2024.3402430}}

@ARTICLE{7405263,
  author={Nedić, Angelia and Olshevsky, Alex},
  journal={IEEE Transactions on Automatic Control}, 
  title={Stochastic Gradient-Push for Strongly Convex Functions on Time-Varying Directed Graphs}, 
  year={2016},
  volume={61},
  number={12},
  pages={3936-3947},
  doi={10.1109/TAC.2016.2529285}}

@INPROCEEDINGS{6426375,
  author={Tsianos, Konstantinos I. and Lawlor, Sean and Rabbat, Michael G.},
  booktitle={2012 IEEE 51st Conference on Decision and Control (CDC)}, 
  title={Push-Sum Distributed Dual Averaging for convex optimization}, 
  year={2012},
  volume={},
  number={},
  pages={5453-5458},
  doi={10.1109/CDC.2012.6426375}}

@ARTICLE{6930814,
  author={Nedić, Angelia and Olshevsky, Alex},
  journal={IEEE Transactions on Automatic Control}, 
  title={Distributed Optimization Over Time-Varying Directed Graphs}, 
  year={2015},
  volume={60},
  number={3},
  pages={601-615},
  doi={10.1109/TAC.2014.2364096}}

@ARTICLE{8988200,
  author={Pu, Shi and Shi, Wei and Xu, Jinming and Nedić, Angelia},
  journal={IEEE Transactions on Automatic Control}, 
  title={Push–Pull Gradient Methods for Distributed Optimization in Networks}, 
  year={2021},
  volume={66},
  number={1},
  pages={1-16},
  doi={10.1109/TAC.2020.2972824}}

@ARTICLE{10337617,
  author={Nguyen, Duong Thuy Anh and Nguyen, Duong Tung and Nedić, Angelia},
  journal={IEEE Transactions on Control of Network Systems}, 
  title={Accelerated $AB$/Push–Pull Methods for Distributed Optimization Over Time-Varying Directed Networks}, 
  year={2024},
  volume={11},
  number={3},
  pages={1395-1407},
  doi={10.1109/TCNS.2023.3338236}}

@ARTICLE{10115431,
  author={Wang, Yongqiang and Nedić, Angelia},
  journal={IEEE Transactions on Automatic Control}, 
  title={Tailoring Gradient Methods for Differentially Private Distributed Optimization}, 
  year={2024},
  volume={69},
  number={2},
  pages={872-887},
  doi={10.1109/TAC.2023.3272968}}

@INPROCEEDINGS{7526803,
  author={Van Sy Mai and Abed, Eyad H.},
  booktitle={2016 American Control Conference (ACC)}, 
  title={Distributed optimization over weighted directed graphs using row stochastic matrix}, 
  year={2016},
  volume={},
  number={},
  pages={7165-7170},
  doi={10.1109/ACC.2016.7526803}}

@ARTICLE{8267245,
  author={Xi, Chenguang and Mai, Van Sy and Xin, Ran and Abed, Eyad H. and Khan, Usman A.},
  journal={IEEE Transactions on Automatic Control}, 
  title={Linear Convergence in Optimization Over Directed Graphs With Row-Stochastic Matrices}, 
  year={2018},
  volume={63},
  number={10},
  pages={3558-3565},
  doi={10.1109/TAC.2018.2797164}}

@ARTICLE{10124282,
  author={Ghaderyan, Diyako and Aybat, Necdet Serhat and Aguiar, A. Pedro and Pereira, Fernando Lobo},
  journal={IEEE Transactions on Automatic Control}, 
  title={A Fast Row-Stochastic Decentralized Method for Distributed Optimization Over Directed Graphs}, 
  year={2024},
  volume={69},
  number={1},
  pages={275-289},
  doi={10.1109/TAC.2023.3275927}}

@ARTICLE{9013030,
  author={Xiong, Yongyang and Xu, Jinming and You, Keyou and Liu, Jianxing and Wu, Ligang},
  journal={IEEE Transactions on Control of Network Systems}, 
  title={Privacy-Preserving Distributed Online Optimization Over Unbalanced Digraphs via Subgradient Rescaling}, 
  year={2020},
  volume={7},
  number={3},
  pages={1366-1378},
  doi={10.1109/TCNS.2020.2976273}}

@ARTICLE{10552083,
  author={Gao, Wang and Zhao, Zhongyuan and Wei, Mengli and Yang, Ju and Zhang, Xiaogang and Li, Jinsong},
  journal={IEEE Transactions on Network Science and Engineering}, 
  title={Decentralized Online Bandit Federated Learning Over Unbalanced Directed Networks}, 
  year={2024},
  volume={11},
  number={5},
  pages={4264-4277},
  doi={10.1109/TNSE.2024.3409755}}

@ARTICLE{11079243,
  author={Zhao, Zhongyuan and Liu, Zhifei and Zhang, Can and Wei, Mengli},
  journal={IEEE Transactions on Vehicular Technology}, 
  title={Differential Privacy Decentralized Federated Learning for Internet of Vehicles over Time-varying Unbalanced Networks}, 
  year={2025},
  volume={},
  number={},
  pages={1-12},
  doi={10.1109/TVT.2025.3585176}}

@ARTICLE{9714350,
  author={Ouadrhiri, Ahmed El and Abdelhadi, Ahmed},
  journal={IEEE Access}, 
  title={Differential Privacy for Deep and Federated Learning: A Survey}, 
  year={2022},
  volume={10},
  number={},
  pages={22359-22380},
  doi={10.1109/ACCESS.2022.3151670}}

@ARTICLE{9069945,
  author={Wei, Kang and Li, Jun and Ding, Ming and Ma, Chuan and Yang, Howard H. and Farokhi, Farhad and Jin, Shi and Quek, Tony Q. S. and Vincent Poor, H.},
  journal={IEEE Transactions on Information Forensics and Security}, 
  title={Federated Learning With Differential Privacy: Algorithms and Performance Analysis}, 
  year={2020},
  volume={15},
  number={},
  pages={3454-3469},
  doi={10.1109/TIFS.2020.2988575}}

@misc{rodio2025optimizingprivacyutilitytradeoffdecentralized,
      title={Optimizing Privacy-Utility Trade-off in Decentralized Learning with Generalized Correlated Noise}, 
      author={Angelo Rodio and Zheng Chen and Erik G. Larsson},
      year={2025},
      eprint={2501.14644},
      archivePrefix={arXiv},
      primaryClass={cs.LG},
      url={https://arxiv.org/abs/2501.14644}, 
}

@misc{abrar2025nonconvexovertheairheterogeneousfederated,
      title={Non-Convex Over-the-Air Heterogeneous Federated Learning: A Bias-Variance Trade-off}, 
      author={Muhammad Faraz Ul Abrar and Nicolò Michelusi},
      year={2025},
      eprint={2510.26722},
      archivePrefix={arXiv},
      primaryClass={cs.LG},
      url={https://arxiv.org/abs/2510.26722}, 
}

@ARTICLE{6930789,
  author={Mateos-Núñez, David and Cortés, Jorge},
  journal={IEEE Transactions on Network Science and Engineering}, 
  title={Distributed Online Convex Optimization Over Jointly Connected Digraphs}, 
  year={2014},
  volume={1},
  number={1},
  pages={23-37},
  doi={10.1109/TNSE.2014.2363554}}

@article{zhu2019deep,
  title={Deep leakage from gradients},
  author={Zhu, Ligeng and Liu, Zhijian and Han, Song},
  journal={Advances in Neural Information Processing Systems},
  volume={32},
  year={2019}
}

@article{geiping2020inverting,
  title={Inverting gradients-how easy is it to break privacy in federated learning?},
  author={Geiping, Jonas and Bauermeister, Hartmut and Dr{\"o}ge, Hannah and Moeller, Michael},
  journal={Advances in Neural Information Processing Systems},
  volume={33},
  pages={16937--16947},
  year={2020}
}

@ARTICLE{10024757,
  author={Geng, Jiahui and Mou, Yongli and Li, Qing and Li, Feifei and Beyan, Oya and Decker, Stefan and Rong, Chunming},
  journal={IEEE Transactions on Big Data}, 
  title={Improved Gradient Inversion Attacks and Defenses in Federated Learning}, 
  year={2024},
  volume={10},
  number={6},
  pages={839-850},
  doi={10.1109/TBDATA.2023.3239116}}

@ARTICLE{9563232,
  author={Xing, Hong and Simeone, Osvaldo and Bi, Suzhi},
  journal={IEEE Journal on Selected Areas in Communications}, 
  title={Federated Learning Over Wireless Device-to-Device Networks: Algorithms and Convergence Analysis}, 
  year={2021},
  volume={39},
  number={12},
  pages={3723-3741},
  doi={10.1109/JSAC.2021.3118400}}

@INPROCEEDINGS{9322286,
  author={Ozfatura, E. and Rini, Stefano and Gündüz, D.},
  booktitle={GLOBECOM 2020 - 2020 IEEE Global Communications Conference}, 
  title={Decentralized {SGD} with Over-the-Air Computation}, 
  year={2020},
  volume={},
  number={},
  pages={1-6},
  doi={10.1109/GLOBECOM42002.2020.9322286}}

@ARTICLE{8316938,
  author={Xie, Pei and You, Keyou and Tempo, Roberto and Song, Shiji and Wu, Cheng},
  journal={IEEE Transactions on Automatic Control}, 
  title={Distributed Convex Optimization with Inequality Constraints over Time-Varying Unbalanced Digraphs}, 
  year={2018},
  volume={63},
  number={12},
  pages={4331-4337},
  doi={10.1109/TAC.2018.2816104}}

@article{MAI201994,
title = {Distributed optimization over directed graphs with row stochasticity and constraint regularity},
journal = {Automatica},
volume = {102},
pages = {94-104},
year = {2019},
issn = {0005-1098},
doi = {https://doi.org/10.1016/j.automatica.2018.07.020},
author = {Van Sy Mai and Eyad H. Abed}
}

@article{dwork2014algorithmic,
  title={The algorithmic foundations of differential privacy},
  author={Dwork, Cynthia and Roth, Aaron and others},
  journal={Foundations and Trends{\textregistered} in Theoretical Computer Science},
  volume={9},
  number={3--4},
  pages={211--407},
  year={2014},
  publisher={Now Publishers, Inc.}
}

@INPROCEEDINGS{11195385,
  author={Ziaeddini, Amir and Yakimenka, Yauhen and Kliewer, Jörg},
  booktitle={2025 IEEE International Symposium on Information Theory (ISIT)}, 
  title={Differentially-Private Decentralized Learning in Heterogeneous Multicast Networks}, 
  year={2025},
  volume={},
  number={},
  pages={1-6},
  doi={10.1109/ISIT63088.2025.11195385}}

@ARTICLE{8952884,
  author={Yang, Kai and Jiang, Tao and Shi, Yuanming and Ding, Zhi},
  journal={IEEE Transactions on Wireless Communications}, 
  title={Federated Learning via Over-the-Air Computation}, 
  year={2020},
  volume={19},
  number={3},
  pages={2022-2035},
  doi={10.1109/TWC.2019.2961673}}

@ARTICLE{11131470,
  author={Xiao, Yue and Ye, Yu and Tegos, Sotiris A. and Diamantoulakis, Panagiotis D. and You, Yang and Karagiannidis, George K.},
  journal={IEEE Transactions on Vehicular Technology}, 
  title={Heterogeneous Wireless Federated Learning Framework via Over-the-Air Computation}, 
  year={2026},
  volume={75},
  number={2},
  pages={3041-3054},
  doi={10.1109/TVT.2025.3600821}}

@ARTICLE{11048956,
  author={Zheng, Zijian and Deng, Yansha and Yi, Wenqiang and Shin, Hyundong and Nallanathan, Arumugam},
  journal={IEEE Transactions on Communications}, 
  title={Over-the-Air Computation Enabled Semi-Asynchronous Wireless Federated Learning}, 
  year={2025},
  volume={73},
  number={10},
  pages={8919-8936},
  doi={10.1109/TCOMM.2025.3582727}}

@inproceedings{Balle2018PrivacyAmplification,
  author    = {Borja Balle and
               Gilles Barthe and
               Marco Gaboardi},
  title     = {Privacy Amplification by Subsampling: Tight Analyses via Couplings and Divergences},
  booktitle = {Adv. in Neural Inf. Process. Syst. 31 (NeurIPS 2018)},
  pages     = {6280--6290},
  year      = {2018},
}

@ARTICLE{10902529,
  author={Huang, Qian and Yi, Xiaoyin and Qi, Fei and Liu, Lei and Xie, Qingming and Jiang, Qin and Hu, Chunxia},
  journal={IEEE Transactions on Broadcasting}, 
  title={Enhancing 5{G} {V2X} {URLLC} Broadcast/Multicast Services With {FL}-Based Wireless Resource Allocation}, 
  year={2025},
  volume={71},
  number={2},
  pages={384-396},
  doi={10.1109/TBC.2025.3541887}}

@ARTICLE{9014530,
  author={Amiri, Mohammad Mohammadi and Gündüz, Deniz},
  journal={IEEE Transactions on Wireless Communications}, 
  title={Federated Learning Over Wireless Fading Channels}, 
  year={2020},
  volume={19},
  number={5},
  pages={3546-3557},
  doi={10.1109/TWC.2020.2974748}}

@INPROCEEDINGS{11148660,
  author={Feng, Hexin and Wang, Rui and Liu, Erwu and Ni, Wei},
  booktitle={2025 IEEE/CIC International Conference on Communications in China (ICCC)}, 
  title={Decentralized Federated Learning via Dynamic Topologies and {MIMO} Over-the-Air Computation}, 
  year={2025},
  volume={},
  number={},
  pages={1-6},
  doi={10.1109/ICCC65529.2025.11148660}}

@article{DBLP:journals/corr/abs-1808-00087,
  author       = {Yu{-}Xiang Wang and
                  Borja Balle and
                  Shiva Prasad Kasiviswanathan},
  title        = {Subsampled {R{\'{e}}nyi} Differential Privacy and Analytical Moments
                  Accountant},
  journal      = {CoRR},
  volume       = {abs/1808.00087},
  year         = {2018},
  url          = {http://arxiv.org/abs/1808.00087},
  eprinttype   = {arXiv},
  eprint       = {1808.00087},
  bibsource    = {dblp computer science bibliography, https://dblp.org}
}

@article{DBLP:journals/corr/Mironov17,
  author       = {Ilya Mironov},
  title        = {Renyi Differential Privacy},
  journal      = {CoRR},
  volume       = {abs/1702.07476},
  year         = {2017},
  url          = {http://arxiv.org/abs/1702.07476},
  eprinttype   = {arXiv},
  eprint       = {1702.07476},
  bibsource    = {dblp computer science bibliography, https://dblp.org}
}

% -------------------------------------------------------------
% \begin{IEEEbiographynophoto}{Amir Ziaeddini}
% Biography...
% \end{IEEEbiographynophoto}

\end{document}